\documentclass[12pt,reqno]{article}
\usepackage[Symbol]{upgreek}
\usepackage[T1]{fontenc}
\usepackage[latin9]{inputenc}
\usepackage[letterpaper]{geometry}
\usepackage{color}
\usepackage{amsmath}
\usepackage{amsthm}
\usepackage{amssymb}
\usepackage{setspace}
\usepackage{natbib}
\usepackage[unicode=true,pdfusetitle,
 bookmarks=true,bookmarksnumbered=false,bookmarksopen=false,
 breaklinks=false,pdfborder={0 0 0},pdfborderstyle={},backref=false,colorlinks=true]
 {hyperref}
\hypersetup{
 linkcolor=magenta,citecolor=blue}
 
\usepackage{graphicx}
\graphicspath{ {./graphs/} }
\usepackage{babel}
\providecommand{\tabularnewline}{\\}
\makeatother
\usepackage{array}
\usepackage{booktabs}
\usepackage{multirow}

\makeatletter

\usepackage{babel}

\makeatletter
\numberwithin{equation}{section}
\numberwithin{figure}{section}
\theoremstyle{remark}
\newtheorem{rem}{\protect\remarkname}[section]
\theoremstyle{definition}
\newtheorem{defn}{\protect\definitionname}[section]
\theoremstyle{plain}
\newtheorem{thm}{\protect\theoremname}[section]
\theoremstyle{plain}

\theoremstyle{definition}
\newtheorem{example}{\protect\examplename}[section]
\theoremstyle{definition}

\newtheorem{lem}{\protect\lemmaname}[section]

\@ifundefined{date}{}{\date{}}
\usepackage{amsmath}
\usepackage{graphicx,psfrag,epsf}
\usepackage{enumerate}
\usepackage{url} 

\usepackage{amsthm}
\usepackage{amsfonts}
\usepackage{amstext}

\usepackage{pdflscape}
\usepackage{threeparttable}

\makeatother

\providecommand{\conditionname}{Condition}
\providecommand{\definitionname}{Definition}
\providecommand{\examplename}{Example}
\providecommand{\lemmaname}{Lemma}
\providecommand{\propositionname}{Proposition}
\providecommand{\remarkname}{Remark}
\providecommand{\theoremname}{Theorem}

\usepackage{todonotes}
\usepackage{verbatim}

\begin{document}
\sloppy

\newgeometry{tmargin=1in,bmargin=1in,lmargin=1.25in,rmargin=1.25in,footskip=1cm}

\title{Treatment Choice, Mean Square Regret and Partial Identification{\thanks{%
The authors gratefully acknowledge financial support from ERC grants (numbers 715940 for Kitagawa and 646917 for Lee), the
ESRC Centre for Microdata Methods and Practice (CeMMAP) (grant number RES-589-28-0001) and the NSF grant (number SES-2315600 for Qiu).}}}

\author{Toru Kitagawa\thanks{ Department of Economics, Brown University and Department of Economics, University College London. Email: toru\_kitagawa@brown.edu} \and Sokbae Lee\thanks{Department of Economics, Columbia University. Email: sl3841@columbia.edu} \and Chen Qiu\thanks{Department of Economics, Cornell University. Email: cq62@cornell.edu}}

\date{September 2023}
\maketitle

\begin{abstract}
We consider a decision maker who faces a binary treatment choice when their welfare is only partially identified from data. We contribute to the literature by anchoring our finite-sample analysis on mean square regret, a decision criterion advocated by \citet*{kitagawa2022treatment}. We find  that optimal rules are always fractional, irrespective of the width of the identified set and precision of its estimate. The optimal treatment fraction is a simple logistic transformation of the commonly used t-statistic multiplied by a factor calculated by a simple constrained optimization. This treatment fraction gets closer to 0.5 as the width of the identified set becomes wider, implying the decision maker becomes more cautious against the adversarial Nature.\\ 

\end{abstract}

\restoregeometry

\newpage{}

\onehalfspacing

\section{Introduction}

Evidence-based policy making has been a keyword  among researchers in social sciences and practitioners of public policies. A central question in evidence-based policy making is: how should a policy maker inform an optimal policy given information gathered from finite data? The seminal work of \cite{manski2004statistical} advocates to approach the question via the framework of a \emph{statistical treatment choice}, where the planner's policy choice is formulated based on the statistical decision theory of \cite{Wald50}.

Ultimately, the selection of an optimal policy depends on the criterion of the decision maker. In the literature of statistical treatment choice, a widely used notion is regret \citep{Savage51}, essentially the sub-optimality welfare gap between a policy under investigation and the oracle first-best policy. Furthermore, a common practice is to select optimal rules via minimax regret, which ranks decision rules via their worst-case expected regret over the  underlying state of nature governing the sampling distribution and causal effects of the policy.

In a setting with point-identified welfare, optimal decision rules based on minimax regret are often singleton rules (e.g., \citealt{stoye2009minimax} and \citealt{tetenov2012statistical}), i.e., they dictate to either treat everyone, or no one in the whole population given realized values of sample data. In a setting with partially-identified welfare, minimax regret optimal rules can be either singleton or non-singleton rules. See, for example, \cite{manski20092009,tetenov2009measuring,stoye2012minimax,yata2021}. Recently, in a point-identified case, \cite*{kitagawa2022treatment} found that singleton rules can be sensitive to the sampling uncertainty and may incur a high chance of large welfare loss (see \citealt{kitagawa2022treatment} for further analyses). As a result, \cite{kitagawa2022treatment} advocate the use of nonlinear regret to rank decision rules. 
For example, \cite{kitagawa2022treatment} recommend using mean square regret as a default, which penalizes rules with large variance of regret. This approach aligns with the choice of a decision maker who displays regret aversion, as axiomatized by \citet{hayashi2008regret}. In a binary treatment setup, \cite{kitagawa2022treatment} show that minimax optimal decision rules with mean square regret are always fractional and follow a simple form of a logistic transformation of the commonly used t-statistic for the welfare contrast.


The particular minimax optimal rules derived in  \cite{kitagawa2022treatment} focus on the case with point-identified welfare. That is, as finite sample data becomes large, the decision maker is able to learn the true welfare of each treatment and thus also to learn the true optimal treatment policy. While this assumption can be satisfied in many scenarios involving experimental data, there are still plenty of situations when such assumptions might be reasonably questioned. For example, even in randomized control trials (RCTs), outcome data under treatment or control might still be missing due to noncompliance of the sample units or due to attrition in the data-collecting process. Even without noncompliance or attrition and when the RCTs are internally valid, researchers may also be concerned about external validity, in the sense that the population for which the treatment policy is applied may be different from the population under which the RCTs are conducted. 

\begin{table}\label{tab:motivation}
\begin{centering}
\begin{tabular}{ccc}
\toprule 
Minimax optimal rule & mean regret & mean square regret\tabularnewline
\midrule
\multirow{1}{*}{point-identified welfare} & singleton & fractional \medskip \tabularnewline
\multirow{1}{*}{partially-identified welfare} & either singleton or fractional & \multirow{1}{*}{\textbf{aim of this paper}}\tabularnewline
\bottomrule
\end{tabular}
\par\end{centering}
\caption{Treatment choice with partial identification: existing results and aim of this paper}
\end{table}

What is the optimal treatment policy when a decision maker cares about mean square regret but faces the problem of a partially-identified welfare?  
Do the results of \cite{kitagawa2022treatment} that optimal rules are fractional remain to hold under partial identification? 
This paper aims to address these questions in a finite-sample framework, extending the analyses by  \cite{kitagawa2022treatment}. See Table \ref{tab:motivation} for an illustration of the motivation of the paper in relation with the existing results in the literature. Following earlier studies by \cite{Manski2000,brock2006profiling,Manski2007,tetenov2009measuring,stoye2012minimax}, among others, we adopt a simple, but well-motivated regret-based framework in which a policy maker, who wishes to maximize the expected outcome of the population, needs to choose a binary treatment when (1) the average treatment effect of the target population is partially identified, but (2) the identified set for the average treatment effect of the target population is a symmetric interval with a fixed and known length around the point-identified reduced-form parameter, for which a Gaussian sufficient statistic is available. Scenarios sharing both or either of the features have been studied by, e.g., \cite{tetenov2009measuring,stoye2012minimax,d2021policy, ishihara2021,yata2021,adjaho2022externally,ben2022policy,kido2022distributionally,christensen2023optimal}.

This paper contributes to the literature by developing new finite-sample optimal decision rules with mean square regret under partial identification, which has not been considered elsewhere in the literature to the best of our knowledge. We show that the fundamental form of the minimax optimal rules derived by \cite{kitagawa2022treatment} is preserved in the partial identification case. With partially identified welfare, minimax optimal rules have the following simple logistic form:
\begin{equation}\label{eq:intro.1}
\frac{\exp\left(2\cdotp a^{*}\cdotp\hat{t} \, \right)}{\exp\left(2\cdotp a^{*}\cdotp\hat{t} \, \right)+1},  
\end{equation}
where $\hat{t}$ is the t-statistic for the point-identified reduced-form parameter (say, the average treatment effect of the experimental population in the RCT), and $a^{\ast}\in(0,1.23)$ is the solution of a simple  constrained optimization problem that depends on the ratio of two key parameters: the width of the identified set $k$,  and the standard deviation $\sigma$ of the estimate of the identified set. In the absence of partial identification, $k=0$ and $a^{\ast}=1.23$, and (\ref{eq:intro.1}) becomes the rule derived by \cite{kitagawa2022treatment}. 

The form of rule (\ref{eq:intro.1}) is consistent with the findings by  \cite{kitagawa2022treatment}: minimax optimal rules with mean square regret are always fractional, irrespective of the magnitude of $k$ and $\sigma$. Moreover, $a^*$ is the center of the identified set under the least favorable prior, and  (\ref{eq:intro.1}) is the posterior probability, under that least favorable prior, that the treatment effect of the target population is positive. Due to partial identification, the location of $a^*$ needs to be calibrated in a case-by-case manner. We show that $a^{\ast}<1.23$, so that the treatment fraction given $\hat{t}>0$ is strictly smaller than that in a point-identified case. Therefore, a direct impact of partial identification on treatment choice is that it further disciplines the planner to be more cautious against the adversarial Nature. That is, optimal decision rules will allocate a larger fraction of the population to the opposite treatment, compared to the point-identified case. 

Our results draw a sharp contrast with the existing results by \cite{stoye2012minimax} and \cite{yata2021}, who derive minimax optimal rules under the same framework but with mean regret.  Firstly, their results show that optimal decision rules are fractional only when $k$ is large enough relative to $\sigma$. 
If $k$ is sufficiently small, minimax regret optimal rules are still singleton rules. With our mean square regret criterion, minimax optimal rules are always fractional. Secondly, if mean regret is the risk function,  whenever a fractional rule is optimal, the corresponding least favorable prior pins down the center of the identified set at a value of 0, i.e., under the least favorable prior, data is  uninformative regarding the sign of the treatment effect of the target population. In contrast, under mean square regret, the least favorable prior for the  center of the identified set supports two points symmetric around 0 so that the decision maker can update that prior with the data.

Due to the set-identified nature of the welfare and the nonlinear nature of the mean square regret, derivation of our results is more delicate than those considered in the existing literature. Indeed, the form of the optimal decision rule depends explicitly on the location of the least favorable prior, which will change depending on the ratio of $k$  and $\sigma$. Following \cite{donoho1994statistical} and \cite{yata2021}, we find our minimax optimal rule by searching for the hardest one-dimensional subproblem and verifying that the minimax optimal rule for the hardest one-dimensional subproblem is indeed minimax optimal for the whole problem. This approach is different from, but very much related to the  \emph{guess-and-verify} approach \citep[as exploited in][among others]{stoye2009minimax,stoye2012minimax,kitagawa2022treatment,azevedo2023b}. As we will demonstrate from Section \ref{sec:main} below, the approach by searching for the one-dimensional subproblem still has a ``guessing'' component as well as a ``verifying'' component. In fact, one may view finding the hardest one-dimensional subproblem as one specific way of figuring out the least favorable prior. Technically, in our considered problem,  one can still try to figure out the structure of the least favorable prior based on prior work (e.g., \citealt{stoye2012minimax}) without using the techniques employed in this paper. Hence, it is not entirely clear which approach has a clear advantage in solving these minimax problems. It is beyond the scope of this paper to investigate optimal rules with mean square regret under the multivariate-signal setting considered by \cite{yata2021}, but we conjecture that similar analyses in this paper may be extended.


Our research is related to a rapidly growing literature on treatment choice with partially identified welfare.  It is known that minimax regret optimal rules may be fractional with or without true knowledge of the identified set \citep{Manski2000,manski2002treatment,manski2005social,brock2006profiling,manski2007identification,Manski2007,tetenov2009measuring,stoye2009partial,stoye2012minimax,cassidy2019tuberculosis,manski2013public,manski2021probabilistic,yata2021}. Fractional rules also arise in a setting with point-identified but nonlinear welfare  \citep{manski2007admissible,manski20092009}. Our results focus on a scenario when the policy maker cannot differentiate each individual in the population. There is also a large literature on individualized policy learning with concerns on partially identified welfare, including issues like distributional robustness, external validity or asymmetric welfare, by, e.g., 
\cite{kallus2018confounding,d2021policy, ishihara2021,adjaho2022externally,ben2021safe,ben2022policy,kido2022distributionally,christensen2023optimal,lei2023policy}. When welfare is point-identified, finite-sample optimal rules are derived in \cite{schlag2006eleven,stoye2009minimax,HiranoPorter2009,tetenov2012statistical,HiranoPorter2020}. Individualised treatment choice with point-identified welfare is considered in \cite{manski2004statistical,BhattacharyaDupas2012,kitagawa2018should,KT21,MT17,AW17}, among others.

The rest of the paper is organised as follows. Section \ref{sec:setup} introduces our setup. Section \ref{sec:main} presents steps to derive our new minimax mean square regret optimal rules via finding the hardest one-dimensional subproblem. Section \ref{sec:conclusion} concludes.

\section{Setup}\label{sec:setup}

Our analysis begins with the basic framework of optimal treatment choice with partially identified welfare and with finite-sample data (see also \citealt{Manski2000,brock2006profiling,Manski2007,manski20092009,tetenov2009measuring,stoye2012minimax} for earlier investigations).
A decision maker contemplates assigning a binary treatment $D\in\{0,1\}$
to an infinitely large population which we call \textit{target} population.
Let $Y_{t}(1)$ be the potential outcome of the target population
when $D=1$ (treatment), and $Y_{t}(0)$ be the potential outcome
of the target population when $D=0$ (control). Denote by $P_{t}\in\mathcal{P}$
the joint distribution of $\left\{ Y_{t}(1),Y_{t}(0)\right\}$. We
assume that a planner aims to maximize the mean outcome of the target population. Define the average treatment
effect of the target population as $\theta_{t}:=\mathbb{E}_{t}\left[Y_{t}(1)-Y_{t}(0)\right]$,
where $\mathbb{E}_{t}[\cdotp]$ denotes the expectation with respect
to $P_{t}$. Then, it is easy to see that the infeasible optimal treatment
policy for the target population is 
\[
\mathbf{1}\left\{ \theta_{t}\geq0\right\} .
\]

To learn about the unknown parameter $\theta_{t}\in\mathbb{R}$, the
decision maker has access to finite data collected from some
RCTs. However, we assume that the RCTs
are implemented on a population, which we call \textit{experimental}
population, that is potentially different from the target population.
That is, the decision maker is concerned about the external validity
of the RCT: the data only has limited validity and the RCTs only partially
identify the true parameter of interest $\theta_{t}$. To derive finite
sample optimality results, we assume that the RCTs have internal validity
so that the decision maker is able to derive a normally distributed estimator $\hat{\theta}_{e}\in\mathbb{R}$
for the average treatment effect of the experimental population. That is,
\[
\hat{\theta}_{e}\sim N(\theta_{e},\sigma^{2}),
\]
where $\theta_{e}\in\mathbb{R}$ is the unknown average treatment
effect of the experimental population, and $\sigma^{2}>0$ is known. Note $\theta_{e}$ is the point-identified reduced-form parameter. 
And $\theta_{e}$ is potentially different
from $\theta_{t}$, which is the parameter of interest that the decision
maker really cares about. Without any assumptions on the relationship
between $\theta_{e}$ and $\theta_{t}$, the problem becomes trivial,
as $\theta_{e}$ and $\theta_{t}$ can be arbitrarily different so
that nothing can be learnt from the RCTs about $\theta_{t}$. In that
sense, data is completely useless. The potential usefulness
of data in revealing the true unknown $\theta_{t}$ lies in the following
key assumption: for each $\theta_{e}\in\mathbb{R}$, the decision
maker knows a priori that the difference between $\theta_{t}$ and
$\theta_{e}$ can be at most $k\in\mathbb{R}$, a known constant. That is, the identified
set for $\theta_{t}$ is:
\begin{equation}
I(\theta_{e}):=[\theta_{e}-k,\theta_{e}+k], \forall\theta_{e}\in \mathbb{R},\label{eq:id.set-2}
\end{equation}
with $k>0$ known. Note the case of $k=0$ corresponds to the
point-identified case in which $\theta_{t}$ and $\theta_{e}$ coincide.
The case of $k=\infty$ corresponds to the case when RCT data is completely
uninformative about the true $\theta_{t}$. 
\begin{rem}
The shape of the identified set $I(\theta_{e})$ in (\ref{eq:id.set-2}) is a symmetric interval around $\theta_{e}$.  Moreover, the upper and lower bounds of $I(\theta_{e})$ are both affine in $\theta_{e}$ with the same gradient. Such a nice structure facilitates finite-sample analysis and arises in many problems, including the missing data \citep{manski1989anatomy}, extrapolation under a Lipshitz assumption \citep{stoye2012minimax,ishihara2021,yata2021}, and welfare bounds with externally invalid experimetnal population \citep{adjaho2022externally,kido2022distributionally}. However, there are also many situations when $I(\theta_{e})$ does not have the nice form in (\ref{eq:id.set-2}). Deriving finite-sample results will be more challenging and is beyond the scope of this paper, and we leave them for future research.   
\end{rem}

The decision maker needs to choose a  statistical treatment rule that maps the empirical evidence summarized by $\hat{\theta}_{e} \in \mathbb{R}$
to the unit interval: 
\[
\hat{\delta}:\mathbb{R}\rightarrow[0,1],
\]
where $\hat{\delta}(x)$ is the fraction of the target population
to be treated after the policy maker observes $\hat{\theta}_{e}=x$.
Note we assume that the primitive action space for the planner is
$[0,1]$. That is, fractional treatment allocation according to some randomization device
is allowed after data have been observed.

We deviate from the existing literature in treatment choice by evaluating
the performance of $\hat{\delta}$ via mean square regret, a decision criterion advocated by \cite{kitagawa2022treatment} as a special case of nonlinear regret. In a setting with point-identified welfare and with finite-sample data, \cite{kitagawa2022treatment} observe that optimal rules are usually singleton rules and are sensitive to the sampling uncertainty. To alleviate concerns regarding the robustness of optimal decision rules with respect to sampling uncertainty, \cite{kitagawa2022treatment} advocate the criteria of nonlinear regret, which incorporates other useful information from the regret distribution (e.g., the second or higher moments), while the standard regret criterion only focuses on the mean of the regret distribution. In particular, mean square regret criterion penalizes rules with large variance of regret, and yields optimal treatment fractions with a simple formula. From the perspective of decision theory, mean square regret also characterizes the choice behaviour of a decision maker who displays regret aversion, a notion axiomatized by \cite{hayashi2008regret}. A natural open question is how the optimal rules will change under the mean square regret criterion if the welfare is now partially identified, which we address in this paper. To proceed, note that applying $\hat{\delta}$ to the target population
yields a welfare of
\begin{align*}
W(\hat{\delta},P_{t}) & :=\hat{\delta}\mathbb{E}_{t}\left[Y_{t}(1)\right]+(1-\hat{\delta})\mathbb{E}_{t}\left[Y_{t}(0)\right]
\end{align*}
and a regret of 
\[
Reg(\hat{\delta},P_{t}):=W(\mathbf{1}\left\{ \theta_{t}\geq0\right\} ,P_{t})-W(\hat{\delta},P_{t})=\theta_{t}\left\{ \mathbf{1}\{\theta_{t}\geq0\}-\hat{\delta}\right\} 
\]
to the planner. The mean square regret of $\hat{\delta}$ is defined
as
\[
R_{sq}(\hat{\delta},\theta_{e},P_{t}):=\mathbb{E}_{\theta_{e}}\left[Reg^{2}(\hat{\delta},P_{t})\right],
\]
where $\mathbb{E}_{\theta_{e}}[\cdotp]$ is with respect to RCT data
$\hat{\theta}_{e}\sim N(\theta_{e},\sigma^{2})$. As $Reg(\hat{\delta},P_{t})$
depends on $P_{t}$ only through $\theta_{t}$, we can simplify $R_{sq}(\hat{\delta},\theta_{e},P_{t})$
as
\[
R_{sq}(\hat{\delta},\theta):=\theta_{t}^{2}\mathbb{E}_{\theta_{e}}\left[\left(\mathbf{1}\{\theta_{t}\geq0\}-\hat{\delta}\right)^{2}\right],
\]
where $\theta:=\left(\begin{array}{c}
\theta_{e}\\
\theta_{t}
\end{array}\right)\in\Theta\subseteq\mathbb{R}^2$ are the unknown parameters in the problem, and
\[
\Theta:=\left\{ (\theta_{e},\theta_{t})^{\prime}\in\mathbb{R}^{2}|\theta_{e}\in\mathbb{R},\theta_{t}\in I(\theta_{e})\right\} 
\]
is the associated parameter space. 

\section{Minimax optimal rules}\label{sec:main}

We aim to find a minimax optimal rule in terms of mean square regret.
Viewing $R_{sq}(\hat{\delta},\theta)$ as the risk function in statistical
decision theory, we introduce the following standard definition of
minimax optimality. 
\begin{defn}
Let $\mathcal{D}$ be a set of statistical decision rules that are functions of $\hat{\theta}_{e}$.
A rule $\hat{\delta}^{*}$ is mean square regret minimax optimal if it is such that
\[
\sup_{\theta\in\Theta}R_{sq}(\hat{\delta}^{*},\theta)=\min_{\hat{\delta}\in\mathcal{D}}\sup_{\theta\in\Theta}R_{sq}(\hat{\delta},\theta).
\]
\end{defn}
Since $\theta\in\Theta$ is a two-dimensional parameter, finding a
minimax optimal rule is more challenging than in a point-identified
case, which can be  viewed as a special case when $\theta_{e}=\theta_{t}$
and the unknown parameter is one-dimensional. That said, note the standard \emph{guess-and-verify} approach \citep[Proposition 4.2, ][]{kitagawa2022treatment} is still valid. In theory, we can still try to figure out a least favorable prior in $\mathbb{R}^2$ and show the Bayes optimal rule with respect to that hypothetical least favorable  prior, say $\hat{\delta}_{\pi}$, is such that
\[
r(\hat{\delta}_{\pi})=\sup_{\theta\in\Theta} R_{sq}(\hat{\delta}_{\pi},\theta),
\]
where $r(\hat{\delta}_{\pi})$ is the Bayes mean square regret of $\hat{\delta}_{\pi}$ under the hypothetical least favorable prior. Here, we take a different, but related approach that was adopted by \cite{yata2021}, who follows \cite{donoho1994statistical} to find a minimax optimal rule by searching for a hardest
one-dimensional subproblem. 
We discuss the connections between these two approaches in Section \ref{sec:subproblem} and Remark \ref{rem:connection.with.donoho}.

Below, we present the core results of this paper. We first
review and extend some existing results in the one-dimensional problem, which
will be useful for the derivation of the minimax optimal rule in one-dimensional subproblem and also for our two-dimensional
problem. 

\subsection{Review of the existing results in one-dimensional problem}\label{sec:review.one.dimension}
\begin{example} [Stylized one-dimensional problem] \label{eg:1}Let $\bar{Y}_{1}\sim N(\tau,1)$
be normally distributed with an unknown mean $\tau\in[-c,c]$ for
some $0<c<\infty$, and a known variance normalized to one, with the
likelihood function 
\begin{equation}
f(\bar{y}_{1}|\tau)=\phi(\bar{y}_{1}-\tau),\forall\bar{y}_{1}\in\mathbb{R},\label{eq:pdf.normal}
\end{equation}
where $\phi(x)$ is the pdf of a standard normal distribution. The
mean square regret of a rule $\hat{\delta}:\mathbb{R}\rightarrow[0,1]$ based on data $\bar{Y}_{1}$
is 
\[
R_{sq}(\hat{\delta},\tau)=\tau^{2}\mathbb{E}\left[\left(\mathbf{1}\{\tau\geq0\}-\hat{\delta}(\bar{Y}_{1})\right)^{2}\right],
\]
where the expectation $\mathbb{E}[\cdotp]$ is with respect to $\bar{Y}_{1}\sim N(\tau,1)$.
\end{example}
\citet[Example 4.1, ][]{kitagawa2022treatment} focus on the general result when $c=\infty$.
The following lemma extends the result of \cite{kitagawa2022treatment} by
allowing $c$ to be bounded and sufficiently small. Let $\rho(a):=\mathbb{E}\left[\left(\frac{1}{\exp\left(2a\bar{Y}_{1}\right)+1}\right)^{2}\right]$,
where the expectation $\mathbb{E}[\cdotp]$ is with respect to $\bar{Y}_{1}\sim N(a,1)$. 
\begin{lem}[Mean square regret minimax rule in a stylized one-dimensional problem]
\label{lem:main.1}In terms of mean square regret, a minimax optimal
rule in Example \ref{eg:1} is
\[
\hat{\delta}^{*}=\begin{cases}
\frac{\exp\left(2\cdotp\tau^{*}\cdotp\bar{Y}_{1}\right)}{\exp\left(2\cdotp\tau^{*}\cdotp\bar{Y}\right)+1}, & \text{if }c\geq\tau^{*},\\
\frac{\exp\left(2\cdotp c\cdotp\bar{Y}_{1}\right)}{\exp\left(2\cdotp c\cdotp\bar{Y}_{1}\right)+1}, & \text{if }c<\tau^{*},
\end{cases}
\]
where $\tau^{*}\approx1.23$ solves $\sup\limits _{\tau\in[0,\infty)}\tau^{2}\rho(\tau)$.
Moreover, the worst-case mean square regret of $\hat{\delta}^{*}$
is
\[
R_{sq}^{*}:=\sup_{\tau\in[-c,c]}R_{sq}(\hat{\delta}^{*},\tau)=\begin{cases}
\left(\tau^{*}\right)^{2}\rho(\tau^{*})\approx0.12, & \text{if }c\geq\tau^{*},\\
c^{2}\rho(c) & \text{if }c<\tau^{*}.
\end{cases}
\]
\end{lem}
\begin{proof}
See Appendix \ref{sec:app.a}.
\end{proof}

\begin{rem}\label{rem:stylized}
The result of Lemma \ref{lem:main.1} implies that  when $c\geq \tau^{\ast}$, minimax optimal decision rule is the same as the one found in \citet[Theorem 4.2, ][]{kitagawa2022treatment}, while the optimal rule differs when $c< \tau^{\ast}$. This result is very intuitive. We know that a global least favorable prior (when $c$ is allowed to be as large as we want) puts equal probabilities on $\tau^{\ast}$ and $-\tau^{\ast}$. If  $c\geq \tau^{\ast}$, the global least favorable prior is always feasible, so the minimax optimal rule must remain the same. If $c< \tau^{\ast}$, the global least favorable prior is no longer feasible. Instead, Lemma \ref{lem:main.1} shows that the constrained least favorable prior  when $c< \tau^{\ast}$ puts equal probabilities on the boundary points $c$ and $-c$, and the minimax optimal rule is the Bayes optimal rule with respect to that constrained least favorable prior.
\end{rem}
\subsection{One-dimensional subproblem}\label{sec:subproblem}


In this and next subsections, we explain in detail how to derive a minimax optimal rule under mean square regret by using the approach taken by \cite{donoho1994statistical} and \cite{yata2021}. The key idea is to find a one-dimensional subproblem (which we know how to solve from results in Section \ref{sec:review.one.dimension}) that is as difficult as the original two-dimensional problem. In this particular example, as the parameter space $\Theta\subseteq \mathbb{R}^2$ is symmetric, it is natural to consider a one-dimensional subproblem in which the parameter space is simply the line connecting two symmetric points around $(0,0)^\prime$ in $\Theta$ (to be formally introduced below). For such one-dimensional subproblem,  we can use Lemma \ref{lem:main.1} to find its minimax optimal rule and the associated worst-case mean square regret. Then, we search among all such one-dimensional subproblems. The one with the largest worst-case mean square regret is our hardest one-dimensional subproblem, and its associated minimax rule is our ``guess'' of the minimax optimal rule for the original two-dimensional problem. A final crucial step is to verify that this candidate minimax rule derived from the hardest one-dimensional subproblem is indeed a minimax rule of the original problem---this corresponds to the ``verifying'' step. Therefore, the approach taken  by \cite{donoho1994statistical} and \cite{yata2021} still has a ``guessing'' component and a ``verifying'' component, and is very much related to the \emph{guess-and-verify} approach that focuses on finding a least favorable prior (exploited in, e.g., \citealt{stoye2009minimax,stoye2012minimax,kitagawa2022treatment,azevedo2023b}). We further discuss the connections between the two approaches in Remark \ref{rem:connection.with.donoho}.

To be more concrete, a one-dimensional subproblem embedded in the two-dimensional problem
can be constructed as follows. Let $a_{e}\geq0$ and $a_{t}\in I(a_{e})$
be two known constants. It follows then $\left(\begin{array}{c}
a_{e}\\
a_{t}
\end{array}\right)\in\Theta$ and $\left(\begin{array}{c}
-a_{e}\\
-a_{t}
\end{array}\right)\in\Theta$. Let 
\[
\Theta_{a_{e},a_{t}}:=\left\{ \theta\in\mathbb{R}^{2}|\theta=s\left(\begin{array}{c}
a_{e}\\
a_{t}
\end{array}\right),s\in[-1,1]\right\} \subseteq\Theta
\]
be the line connecting $\left(\begin{array}{c}
a_{e}\\
a_{t}
\end{array}\right)$ and $\left(\begin{array}{c}
-a_{e}\\
-a_{t}
\end{array}\right)$. The parameter space $\Theta_{a_{e},a_{t}}$ is one-dimensional as
it contains only one unknown parameter $s\in[-1,1]$. We call the
problem of finding a minimax optimal rule when $\theta\in\Theta_{a_{e},a_{t}}$
a \emph{one-dimensional subproblem}. 
Indeed, for intuition, suppose $a_{e}>0$ and let $\hat{s}:=\frac{\hat{\theta}_{e}}{a_{e}}$. 
Simple algebra shows that
\[
\hat{s}\sim N\left(s,\frac{\sigma^{2}}{a_{e}^{2}}\right),
\]
which further implies that
\[
a_{t}\hat{s}=\frac{a_{t}}{a_{e}}\hat{\theta}_{e}\sim N\left(sa_{t},\left(\frac{a_{t}}{a_{e}}\right)^{2}\sigma^{2}\right).
\]
That is, $a_{t}\hat{s}$ is normally distributed with an unknown mean 
$sa_{t}$ (since $s$ is unknown) and  with a known
variance $\left(\frac{a_{t}}{a_{e}}\right)^{2}\sigma^{2}$. Note that $sa_{t}$ is the average treatment effect of the target population. We may then apply
Lemma \ref{lem:main.1} to characterize a minimax optimal rule for the one-dimensional
subproblem. The case when $\theta_{e}=0$, in contrast, requires a separate consideration, as this corresponds to the case when data $\hat{\theta}_{e}\sim N(0,\sigma^2)$ reveals no information regarding $s$. See Remark \ref{rem:subproblem} for further discussions. Considering both  cases when $\theta_{e}>0$ and $\theta_{e}=0$, we have the following lemma. 
\begin{lem}[Mean square regret minimax rule of a one-dimensional subproblem]
\label{lem:main.2}A minimax optimal rule for the one-dimensional
subproblem is
\[
\hat{\delta}_{a_{e},a_{t}}^{*}=\begin{cases}
\frac{\exp\left(2\cdotp\tau^{*}\cdotp\frac{a_{t}}{\left|a_{t}\right|\sigma}\hat{\theta}_{e}\right)}{\exp\left(2\cdotp\tau^{*}\cdotp\frac{a_{t}}{\left|a_{t}\right|\sigma}\hat{\theta}_{e}\right)+1}, & \frac{a_{e}}{\sigma}\geq\tau^{*},\\
\\
\frac{\exp\left(2\cdotp\frac{a_{e}}{\sigma}\frac{a_{t}}{\left|a_{t}\right|\sigma}\hat{\theta}_{e}\right)}{\exp\left(2\cdotp\frac{a_{e}}{\sigma}\frac{a_{t}}{\left|a_{t}\right|\sigma}\hat{\theta}_{e}\right)+1}, & 0\leq\frac{a_{e}}{\sigma}<\tau^{*}.
\end{cases}
\]
That is, 
\[
\sup_{\theta\in\Theta_{a_{e},a_{t}}}R_{sq}(\hat{\delta}_{a_{e},a_{t}}^{*},\theta)=\min_{\hat{\delta}\in\mathcal{D}}\sup_{\theta\in\Theta_{a_{e},a_{t}}}R_{sq}(\hat{\delta},\theta).
\]
Moreover, the worst-case mean square regret of $\hat{\delta}_{a_{e},a_{t}}^{*}$
is 
\[
\sup_{\theta\in\Theta_{a_{e},a_{t}}}R_{sq}(\hat{\delta}_{a_{e},a_{t}}^{*},\theta)=\begin{cases}
\frac{a_{t}^{2}\sigma^{2}}{a_{e}^{2}}\left(\tau^{*}\right)^{2}\rho(\tau^{*}), & \frac{a_{e}}{\sigma}\geq\tau^{*},\\
a_{t}^{2}\rho\left(\frac{a_{e}}{\sigma}\right), & 0\leq\frac{a_{e}}{\sigma}<\tau^{*}.
\end{cases}
\]
\end{lem}
\begin{proof}
See Appendix \ref{sec:app.a}.
\end{proof}

\begin{rem}
The interpretation of the minimax optimal rule in the one-dimensional subproblem is as follows. Intuitively, note as long as $a_{e}\neq0$, $\frac{a_{t}}{\left|a_{t}\right|\sigma}\hat{\theta}_{e}:=\hat{t}$ is a standard t-statistic. Consistent with the conclusion from \cite{kitagawa2022treatment}, a minimax optimal rule in this parametric problem is a logistic transformation of $\hat{t}$. If $\frac{a_{e}}{\sigma}\geq\tau^*$, then the minimax optimal rule is a logistic transformation of $2\tau^*\hat{t}$. If, in contrast, $0<\frac{a_{e}}{\sigma}<\tau^*$, then the minimax optimal rule is a logistic transformation of $2\frac{a_{e}}{\sigma}\hat{t}$. As we can see, if $\hat{t}>0$, the treatment fraction when $0<\frac{a_{e}}{\sigma}<\tau^*$ is smaller than the case when $\frac{a_{e}}{\sigma}\geq\tau^*$. Such a structure has intuitive implications on the minimax optimal rule derived later. See Remark \ref{rem:cautious} for a further discussion.
\end{rem}

\begin{rem}\label{rem:subproblem}
The situation when $a_{e}=0$ is particularly interesting and demonstrates further difference between the criterion of mean square regret and that of mean regret. If it holds $a_{e}=0$, then $\hat{\theta}_{e}\sim N(0,\sigma^{2})$. That is, data is completely uninformative and  reveals no information
regarding the unknown $s$. In this situation, $\theta_{t}\in[-|a_{t}|,|a_{t}|]$. This subproblem coincides with what was analyzed by \cite{manski2007identification}. If the mean of the regret is the criterion, \cite{manski2007identification} shows that any rule $\hat{\delta}$ such that $\mathbb{E}[\hat{\delta}(\hat{\theta}_{e})]=\frac{1}{2}$ is a minimax optimal rule, where the expectation is with respect to $\hat{\theta}_{e}\sim N(0,\sigma^{2})$. That is, there are many minimax optimal rules for this particular subproblem. Using the uninformative data can still be minimax optimal under mean regret criterion, as using random data may be purely utilized as a radomization device without affecting the mean of regret. This draws a sharp contrast with mean square regret, under which the  minimax optimal rule is $\hat{\delta}^*_{0,a_{t}}=\frac{1}{2}$. That is,  the minimax optimal rule under mean square regret is to not use data at all and allocate a fraction of $\frac{1}{2}$ of the whole population to treatment. Such a fractional rule may be implemented via a randomization device that does not depend on data. This is intuitively easy to understand: any other rule that (1) is optimal in terms of the mean of regret and (2) uses random data and generates a positive variance of regret is not optimal in terms of mean square regret as they introduce further variance with respect to data without decreasing the mean of regret. 

\end{rem}

\subsection{Hardest one-dimensional subproblem}

From Lemma \ref{lem:main.2}, we see that for each one-dimensional
subproblem where $\theta\in\Theta_{a_{e},a_{t}}$, the worst mean
square regret of the minimax optimal rule depends on the value
of $a_{e}$ and $a_{t}$, both of which are assumed to be known. Let $a_{e}^{*}\geq0$ and $a_{t}^{*}\in I(a_{e}^{*})$ be two constants. We call the problem of finding a minimax optimal rule when $\theta\in\Theta_{a_{e}^{*},a_{t}^{*}}$
the \emph{hardest one-dimensional subproblem} if 
\[
\sup_{\theta\in\Theta_{a_{e}^{*},a_{t}^{*}}}R_{sq}(\hat{\delta}_{a_{e}^{*},a_{t}^{*}}^{*},\theta)=\sup_{a_{e}\geq0,a_{t}\in I(a_{e})}\sup_{\theta\in\Theta_{a_{e},a_{t}}}R_{sq}(\hat{\delta}_{a_{e},a_{t}}^{*},\theta).
\]
That is, $\Theta_{a_{e}^{*},a_{t}^{*}}$ is the one-dimensional parameter
space that yields the largest possible worst-case mean square regret
of its associated minimax rule. If we view the minimax problem as a game between the adversarial Nature and the econometrician, then the hardest one-dimensional subproblem is the problem that the Nature will pick, provided that the Nature is restricted to choose only among the one-dimensional subproblems. 
To characterise the hardest one-dimensional subproblem,
let 
\begin{equation}\label{eq:a.star}
a^{*}\in\arg\sup_{0\leq\tilde{a}_{e}\leq\tau^*} 
\left( \tilde{a}_{e}+\frac{k}{\sigma} \right)^{2} \rho\left(\tilde{a}_{e}\right)  
\end{equation}

\begin{lem}
\label{lem:main.3}[Mean square regret minimax rule of the hardest one-dimensional subproblem]
\begin{itemize}
\item[(i)]  

The hardest one dimensional subproblem corresponds
to $a_{e}^{*}=a^{*}\sigma$, and $a_{t}^{*}=a^{*}\sigma+k$. Let $\Theta_{\mathrm{H}}:=\Theta_{a^{*}\sigma,a^{*}\sigma+k}$
be the hardest one-dimensional parameter space. The minimax optimal rule with
respect to this hardest one dimensional subproblem is
\[
\hat{\delta}_{\text{H}}^{*}:=\hat{\delta}_{a^{*}\sigma,a^{*}\sigma+k}^{*}=\frac{\exp\left(2\cdotp a^{*}\cdotp\frac{\hat{\theta}_{e}}{\sigma}\right)}{\exp\left(2\cdotp a^{*}\cdotp\frac{\hat{\theta}_{e}}{\sigma}\right)+1},
\]
and 
\[\sup_{\theta\in\Theta_{\mathrm{H}}}R_{sq}(\hat{\delta}_{\text{H}}^{*},\theta)=\sigma^{2}\left(a^{*}+\frac{k}{\sigma}\right)^{2}\rho\left(a^{*}\right).\] 
\item[(ii)] $0<a^{*}<\tau^{*}$. 
\item[(iii)] $a^*$ is strictly decreasing in $k$ and strictly increasing in $\sigma$. 
\end{itemize}
\end{lem}
\begin{proof}
See Appendix \ref{sec:app.a}.
\end{proof}
It turns out $\hat{\delta}_{\text{H}}^{*}$ is not only a minimax
optimal rule of the hardest one-dimensional subproblem, but also a minimax
optimal rule of the original two-dimensional problem. That is, choosing the hardest one-dimensional subproblem is still the adversarial Nature's best move, even if they are allowed to choose any parameter in the two-dimensional parameter space.
\begin{thm}\label{thm:1}
\label{thm:main.1}$\sup_{\theta\in\Theta}R_{sq}(\hat{\delta}_{\mathrm{H}}^{*},\theta)=\min_{\hat{\delta}\in\mathcal{D}}\sup_{\theta\in\Theta}R_{sq}(\hat{\delta},\theta).$
That is, $\hat{\delta}_{\mathrm{H}}^{*}$ is a minimax optimal rule in terms of mean square regret for the original two-dimensional problem analyzed in Section \ref{sec:setup}.
\end{thm}
\begin{proof}
See Appendix \ref{sec:app.a}.
\end{proof}
\begin{rem}\label{rem:connection.with.donoho}
By now, we can see a clear connection between the approach taken by \cite{donoho1994statistical} and \cite{yata2021} in finding minimax optimal decisions and the  \emph{guess-and-verify} approach  \citep[Proposition 4.2, ][]{kitagawa2022treatment}. Intuitively, we can view finding the
hardest one-dimensional subproblem as one way of finding the least
favorable  prior. Indeed, in the original two-dimensional problem, the least favorable  prior can be verified to be supported
on $\left(\begin{array}{c}
a^{*}\sigma\\
a^{*}\sigma+k
\end{array}\right)$ and $\left(\begin{array}{c}
-a^{*}\sigma\\
-a^{*}\sigma-k
\end{array}\right)$ with equal probabilities. Technically, once an econometrician figures out the structure of the least favorable  prior (which is possible given prior work in the literature, e.g., \citealt{stoye2012minimax}), they can proceed without using the techniques employed in this paper, by directly invoking \citet[Proposition 4.2, ][]{kitagawa2022treatment}. Therefore, it is not entirely clear which approach has a relative advantage in solving these minimax problems.
\end{rem}

\begin{rem}[Comparison with \citealt{kitagawa2022treatment}]\label{rem:cautious}

If the treatment effect of the target population is point-identified ($k=0$), the theory of  \cite{kitagawa2022treatment} applies and the minimax optimal rule is
$\hat{\delta}^{*}=\frac{\exp\left(2\cdotp \tau^{*}\cdotp\frac{\hat{\theta}_{e}}{\sigma}\right)}{\exp\left(2\cdotp \tau^{*}\cdotp\frac{\hat{\theta}_{e}}{\sigma}\right)+1}$, which agrees with the conclusion from Theorem \ref{thm:1} by mechanically setting $k=0$. Theorem \ref{thm:1} clearly demonstrates  the effect of partial identification ($k>0$) on the optimal decision rules. Partial identification moves the worst-case location  of the point-identified parameter $\theta_{e}$ further toward zero and away from $\tau^*$: the minimax optimal rule becomes $\hat{\delta}^{*}_{\text{H}}=\frac{\exp\left(2\cdotp a^{*}\cdotp\frac{\hat{\theta}_{e}}{\sigma}\right)}{\exp\left(2\cdotp a^{*}\cdotp\frac{\hat{\theta}_{e}}{\sigma}\right)+1}$ with $a^*<\tau^*$. 
Therefore, partial identification further encourages the decision maker to be more cautious against the adversarial Nature: optimal treatment fraction under partial identification will be closer to 0 compared to a point-identified situation. From Lemma \ref{lem:main.3}(iii), we know the value of $a^*$ decreases as $k$ becomes larger: more partial identification results in more ambiguity, leading to more prudent or cautious treatment allocation. If $k=\infty$, then $a^*=0$ and the optimal treatment rule becomes $\hat{\delta}^{*}_{\text{H}}=\frac{1}{2}$. 
   
\end{rem}

\begin{rem}[Comparison with \citealt{stoye2012minimax} amd \cite{yata2021}]\label{rem:comparison.stoye.yata}
The conclusion of Theorem \ref{thm:1} is quantitatively and qualitatively different from the conclusion of \cite{stoye2012minimax} and \cite{yata2021}, who both use the mean of regret as a risk criterion and derive optimal fractional rules when $k$ is large enough. As shown by \cite{stoye2012minimax} (and was generalized by \cite{yata2021} to setups with multivariate signals), if mean of the regret is the risk criterion, whether or not a minimax optimal rule is fractional depends on the magnitude of $k$. 
If $k\leq \sqrt{\frac{\pi}{2}}\sigma$, the naive empirical success rule $\mathbf{1}\{\hat{\theta}_{e}\geq0\}$ is minimax optimal. When $k>\sqrt{\frac{\pi}{2}}\sigma$, a minimax optimal rule is found to be fractional and admits $\hat{\delta}^*=\Phi\bigl(\hat{\theta}_{e}/\sqrt{2k^{2}/\pi-\sigma^{2}}\bigr)$, under which the worst-case location for $\theta_{e}$ is at $0$, i.e.,  when data are uninformative. 
Theorem \ref{thm:1} draws a very different picture compared to the existing literature: first of all, optimal rules are always fractional, irrespective of the magnitude of $k$. Second, the worst-case location for $\theta_{e}$ is at  $\pm a^*\sigma\neq0$, which implies that  data is still informative regarding the true unidentified treatment effect of the target population. See Figure \ref{fig:1} for an illustration of the minimax optimal rules in terms of mean regret and mean square regret with respect to different values of $k$. 


\end{rem}

\begin{figure}[http]
    \centering
    \includegraphics[scale=0.403]{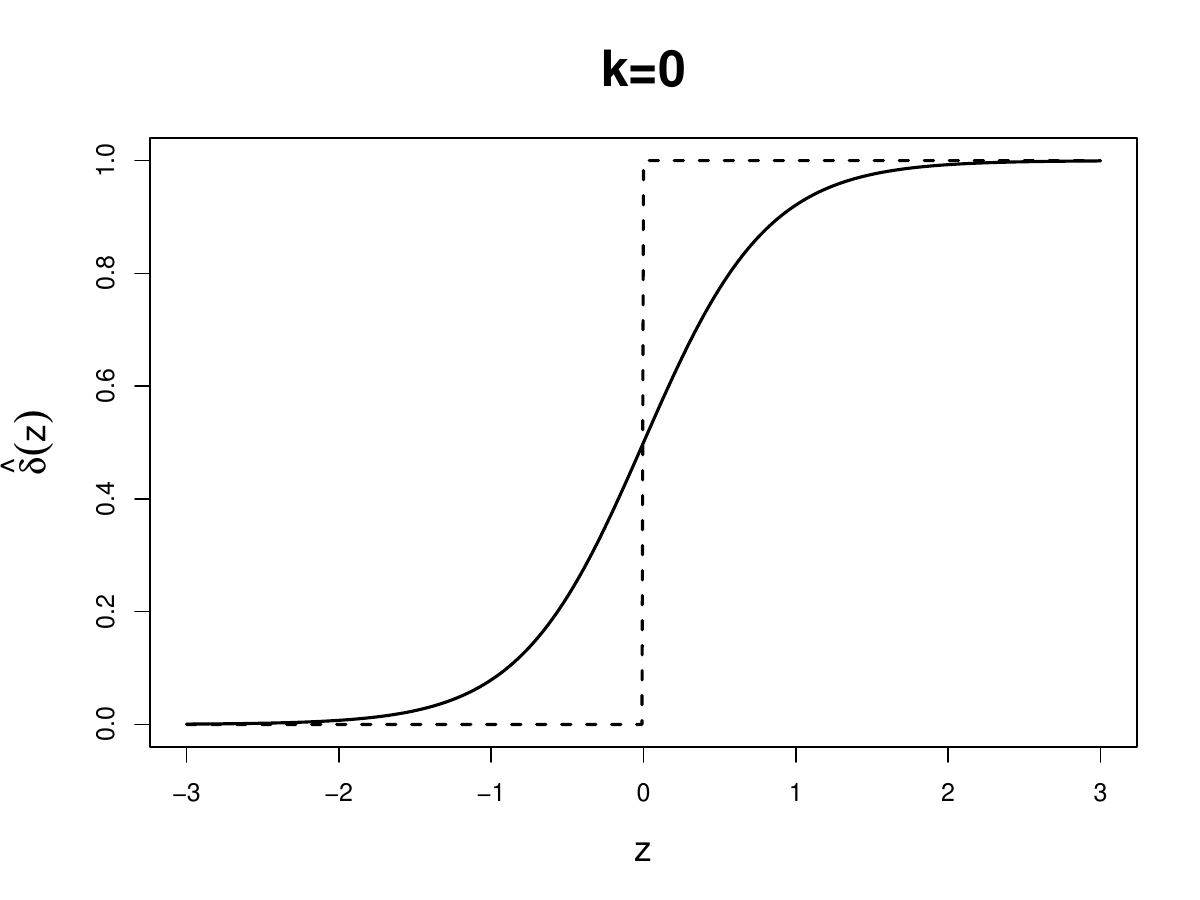}
     \includegraphics[scale=0.403]{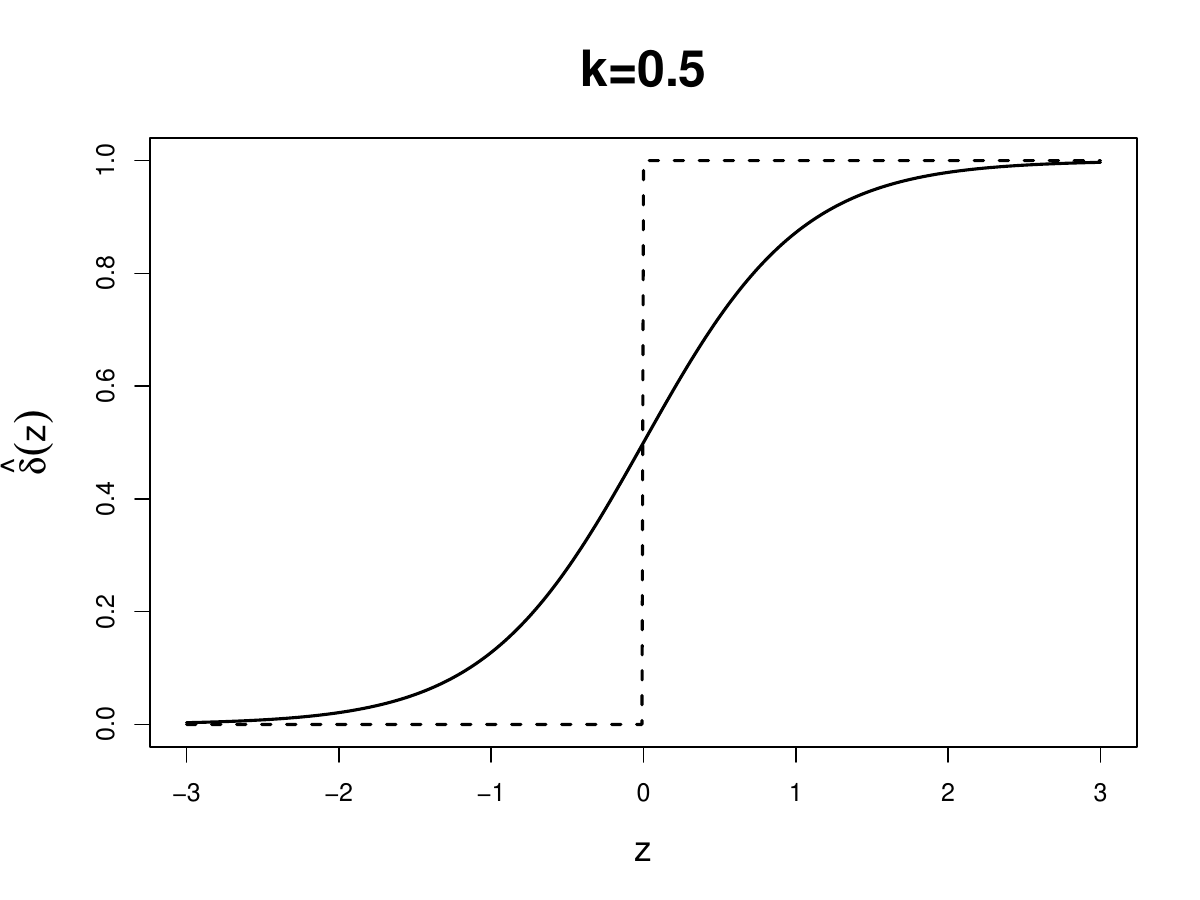}\\
      \includegraphics[scale=0.403]{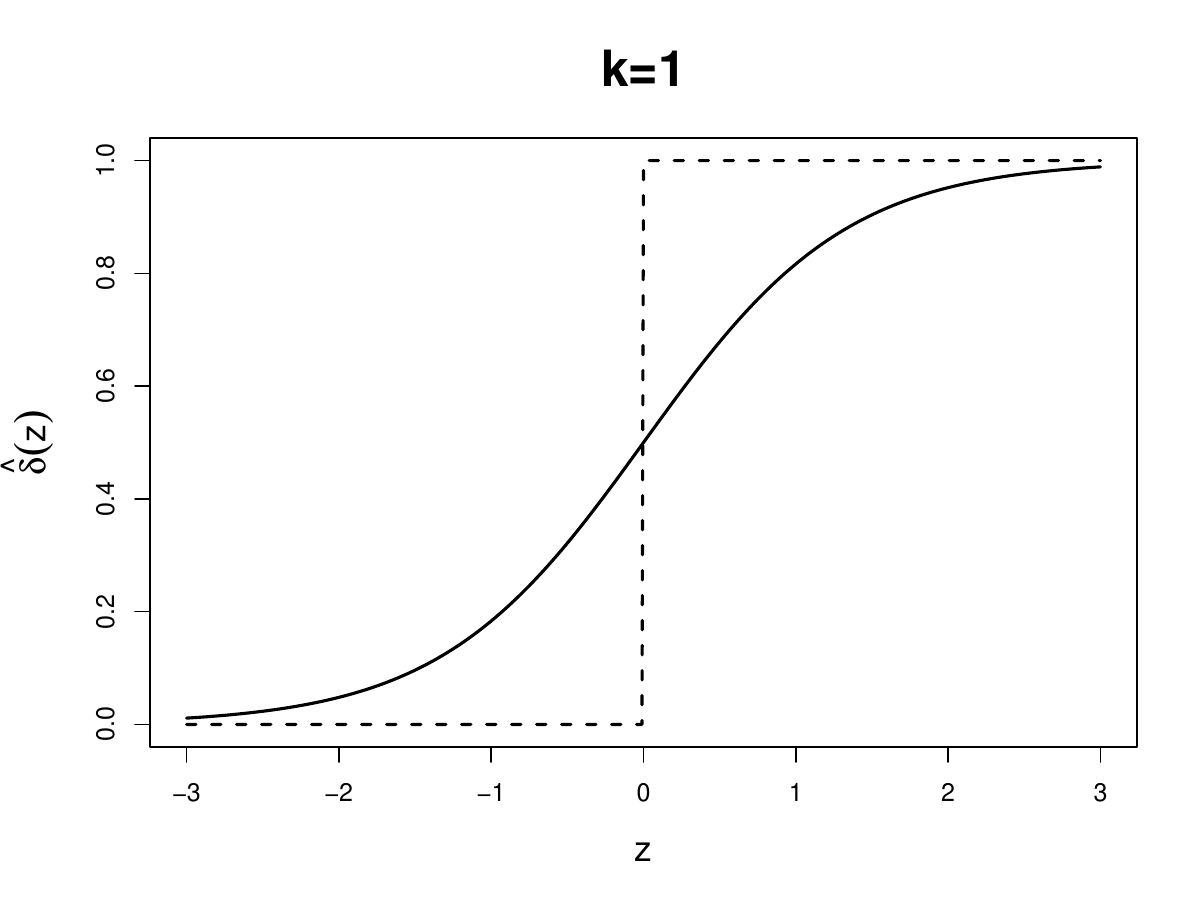}
     \includegraphics[scale=0.403]{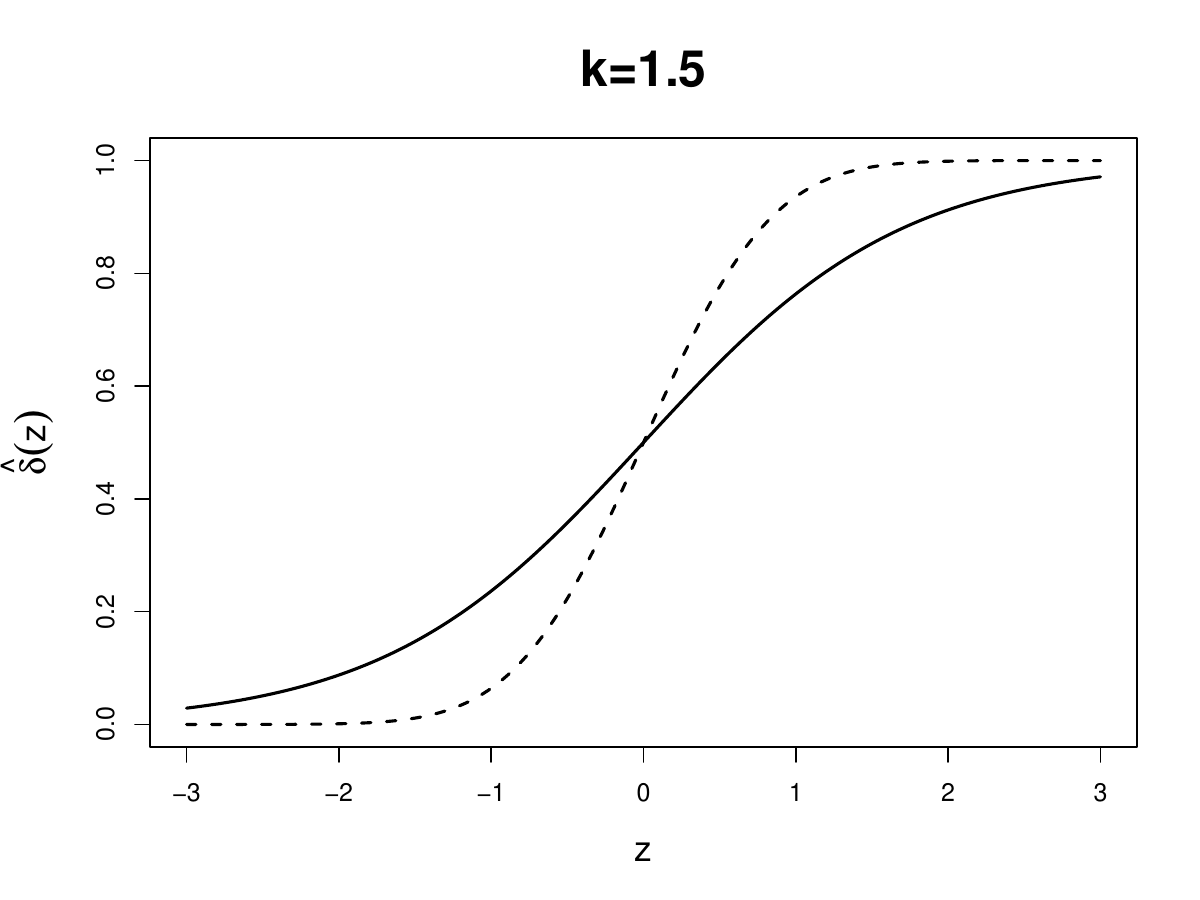}\\
    \includegraphics[scale=0.403]{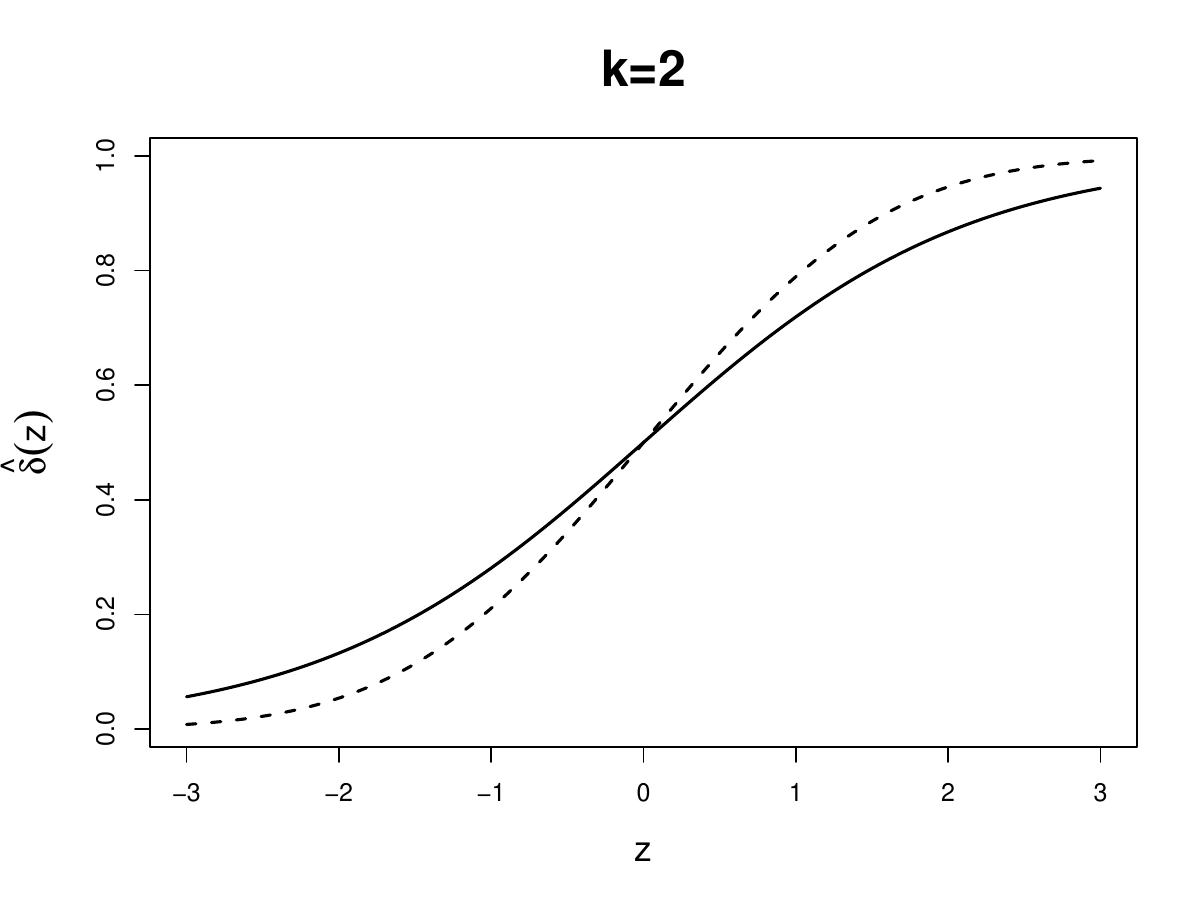}
     \includegraphics[scale=0.403]{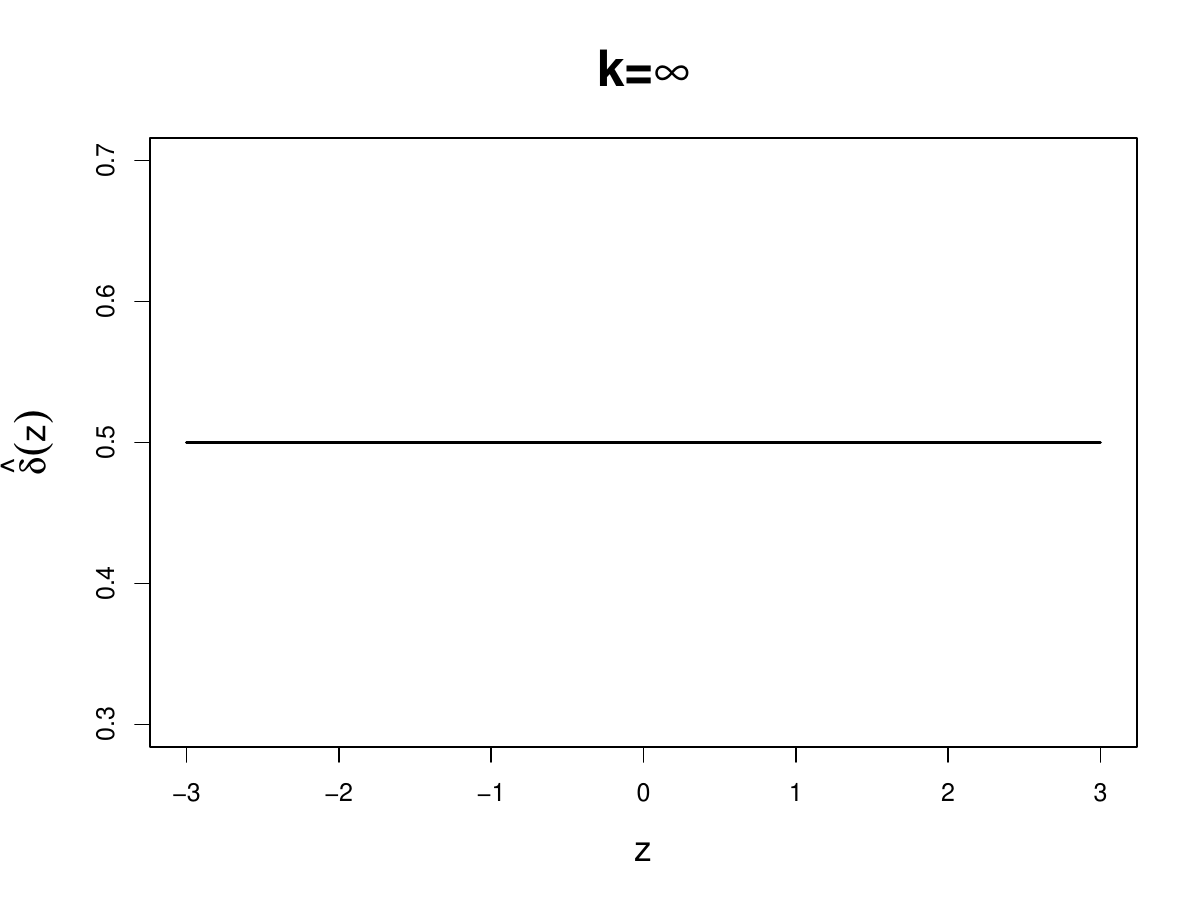}\\
    \caption{Minimax optimal rules in the Gaussian experiment with a unit variance and an unknown mean. In each of the graphs, $k$ represents the width of the identified set. The dashed line is minimax optimal rule with respect to mean regret as a function of $z$, where $z$ represents each possible realization of the Gaussian experiment. The solid line is minimax optimal rule with respect to mean square regret as a function of $z$.  Note in the limiting case $k=\infty$, the two rules coincide.}
    \label{fig:1}
\end{figure}

\section{Conclusion}\label{sec:conclusion}
In this paper, we study optimal binary treatment choice with mean square regret and with partially identified welfare, extending the analyses by \cite{kitagawa2022treatment}. Our results lead to a simple and intuitive rule that is sharply different from the existing literature on treatment choice under partial identification with mean regret criterion. In particular, minimax optimal rules are always fractional, irrespective of the width of the identified set. The optimal treatment fraction is a logistic transformation of the commonly used t-statistic multiplied by a factor that is calculated by a simple constrained optimization. Our results are useful for policy makers who wish to make fractional treatment assignment but are concerned that the true optimal policy can not be identified from data. For future research, it would be interesting to consider optimal treatment choice with a general and arbitrary identified set, or with an estimated identified set. It would also be interesting to consider optimal individualised treatment choice with mean square regret.

\newpage{}

\appendix

\section{Proofs of main results}\label{sec:app.a}

\subsection*{Proof of Lemma \ref{lem:main.1}}

By Remark \ref{rem:stylized}, we focus on
the case when $c<\tau^{*}$. Let $\pi_{c}$ be a prior on $\tau$
such that $\pi_{c}(c)=\pi_{c}(-c)=\frac{1}{2}$. It can be verified that the Bayes optimal rule with respect to $\pi_{c}$ is
\[
\hat{\delta}_{\pi_{c}}(\bar{Y})=\frac{\exp\left(2\cdotp c\cdotp\bar{Y}\right)}{\exp\left(2\cdotp c\cdotp\bar{Y}\right)+1}=\hat{\delta}^{*}(\bar{Y}).
\]
By applying integration by change-of-variable, we may find the Bayes mean square regret of $\hat{\delta}_{\pi_{c}}(\bar{Y})$ as 
\begin{align*}
r_{sq}(\hat{\delta}_{\pi_{c}},\pi_{c}) & :=\int R_{sq}(\hat{\delta}_{\pi_{c}},\tau)d\pi_{c}(\tau)\\
 & =\frac{1}{2}R_{sq}(\hat{\delta}_{\pi_{c}},c)+\frac{1}{2}R_{sq}(\hat{\delta}_{\pi_{c}},-c)\\
 & =c^{2}\rho(c).
\end{align*}
By Lemma \ref{lem:8}, $\sup_{\tau\in[-c,c]}R_{sq}(\hat{\delta}_{\pi_{c}},\tau)=c^{2}\rho(c)$,
implying $\hat{\delta}_{\pi_{c}}$ is indeed a minimax optimal rule
by applying \citet[Proposition 4.2, ][]{kitagawa2022treatment}. 

\subsection*{Proof of Lemma \ref{lem:main.2}}

We prove the lemma by considering two cases.

Case 1: $a_{e}=0$. In this case, for each $\theta\in\Theta_{0,a_{t}}$,
\[
R_{sq}(\hat{\delta},\theta)=\left(a_{t}s\right)^{2}\mathbb{E}\left[\left(\mathbf{1}\{a_{t}s\geq0\}-\hat{\delta}(\hat{\theta}_{e})\right)^{2}\right],
\]
where $\hat{\theta}_{e}\sim N(0,\sigma^{2})$. This is a case where
data $\hat{\theta}_{e}$ reveals no information
regarding the unknown $s$. If in addition to $a_{e}=0$, it holds that $a_{t}=0$. Then, any rule is minimax optimal. Focus on the case when $a_{t}\neq0$. Let $\mu_{\hat{\delta}}:=\mathbb{E}\hat{\delta}(\hat{\theta}_{e})$,
$V_{\hat{\delta}}:=\mathbb{E}\left[\left(\hat{\delta}(\hat{\theta}_{e})-\mathbb{E}\hat{\delta}(\hat{\theta}_{e})\right)^{2}\right]$.
We have the following decomposition
\begin{align*}
R_{sq}(\hat{\delta},\theta) & =\left(a_{t}s\right)^{2}\left\{ \left(\mathbf{1}\{a_{t}s\geq0\}-\mu_{\hat{\delta}}\right)^{2}+V_{\hat{\delta}}\right\} .
\end{align*}
That is, the mean square regret of each rule depends on $\hat{\delta}$
only via $\mu_{\hat{\delta}}$ and $V_{\hat{\delta}}$, both of which
are independent of $s$. Thus, for each $\hat{\delta}$
\begin{align*}
\sup_{\theta\in\Theta_{0,a_{t}}}R_{sq}(\hat{\delta},\theta) & =\max\left\{ a_{t}^{2}\left[\left(1-\mu_{\hat{\delta}}\right)^{2}+V_{\hat{\delta}}\right],a_{t}^{2}\left[\left(\mu_{\hat{\delta}}\right)^{2}+V_{\hat{\delta}}\right]\right\} \\
 & =a_{t}^{2}\left[\max\left(\left(1-\mu_{\hat{\delta}}\right)^{2},\mu_{\hat{\delta}}^{2}\right)+V_{\hat{\delta}}\right].
\end{align*}

As $a_{t}\neq0$, it is easy to see that a minimax optimal rule would set $V_{\hat{\delta}}=0$
and $\mu_{\hat{\delta}}=\frac{1}{2}$. That is, $\hat{\delta}_{0,a_{t}}^{*}=\frac{1}{2}$,
which means that the minimax optimal rule does not use data $\hat{\theta}_{e}$
at all. Moreover, $\sup_{\theta\in\Theta_{0,a_{t}}}R_{sq}(\hat{\delta}_{0,a_{t}}^{*},\theta)=\frac{a_{t}^{2}}{4}$. 

Case 2: $a_{e}>0$. In this case, note for each $\theta\in\Theta_{a_{e},a_{t}}$,
\[
R_{sq}(\hat{\delta},\theta)=\left(a_{t}s\right)^{2}\mathbb{E}_{sa_{e}}\left[\left(\mathbf{1}\{a_{t}s\geq0\}-\hat{\delta}(\hat{\theta}_{e})\right)^{2}\right],
\]
where $\hat{\theta}_{e}\sim N(a_{e}s,\sigma^{2})$. If $a_{t}=0$, then any rule is minimax optimal. Focus on $a_{t}\neq0$. Then, it follows
\[
\frac{a_{t}}{a_{e}}\hat{\theta}_{e}\sim N\left(sa_{t},\left(\frac{a_{t}}{a_{e}}\right)^{2}\sigma^{2}\right).
\]

In this one-dimensional subproblem, $\frac{a_{t}}{a_{e}}\hat{\theta}_{e}$
is a sufficient staitistic for $s$ (and for $a_{t}s$ too). Therefore,
to show $\sup_{\theta\in\Theta_{a_{e},a_{t}}}R_{sq}(\hat{\delta}_{a_{e},a_{t}}^{*},\theta)=\min_{\hat{\delta}\in\mathcal{D}}\sup_{\theta\in\Theta_{a_{e},a_{t}}}R_{sq}(\hat{\delta},\theta)$,
it suffices to focus on rules that are functions of the statistic $\frac{a_{t}}{a_{e}}\hat{\theta}_{e}$
and show 
\[
\sup_{\theta\in\Theta_{a_{e},a_{t}}}R_{sq}(\hat{\delta}_{a_{e},a_{t}}^{*},\theta)=\min_{\hat{\delta}\in\tilde{\mathcal{D}}}\sup_{\theta\in\Theta_{a_{e},a_{t}}}R_{sq}(\hat{\delta},\theta),
\]
where $\mathcal{\tilde{D}}$ is a set of rules that is a function
of the statistic $\frac{a_{t}}{a_{e}}\hat{\theta}_{e}$. 
To this end, let $\tau_{s}:=sa_{t}\in\left[-\left|a_{t}\right|,\left|a_{t}\right|\right]$,
and let $\hat{\tau}_{s}:=\frac{a_{t}}{a_{e}}\hat{\theta}_{e}$. Then,
for each $\hat{\delta}\in\tilde{\mathcal{D}}$ and each $\theta\in\Theta_{a_{e},a_{t}}$,
we can write 
\[
R_{sq}(\hat{\delta},\theta)=\tau_{s}^{2}\mathbb{E}\left[\left(\mathbf{1}\{\tau_{s}\geq0\}-\hat{\delta}(\hat{\tau}_{s})\right)^{2}\right]
\]
where the $\mathbb{E}[\cdotp]$ is with respect to $\hat{\tau}_{s}\sim N\left(\tau_{s},\sigma_{\tau_{s}}^{2}\right),$
where $\sigma_{\tau_{s}}^{2}=\left(\frac{a_{t}}{a_{e}}\right)^{2}\sigma^{2}$.
Furthermore, note
\begin{align}
R_{sq}(\hat{\delta},\theta) & =\tau_{s}^{2}\int\left(\mathbf{1}\{\tau_{s}\geq0\}-\hat{\delta}(x)\right)^{2}\frac{1}{\sigma_{\tau_{s}}}\phi\left(\frac{x-\tau_{s}}{\sigma_{\tau_{s}}}\right)dx\nonumber \\
 & =\sigma_{\tau_{s}}^{2}\left(\frac{\tau_{s}}{\sigma_{\tau_{s}}}\right)^{2}\int\left(\mathbf{1}\left\{ \frac{\tau_{s}}{\sigma_{\tau_{s}}}\geq0\right\} -\hat{\delta}(\sigma_{\tau_{s}}z)\right)^{2}\phi\left(z-\frac{\tau_{s}}{\sigma_{\tau_{s}}}\right)d\left(z\right)\nonumber \\
 & =\sigma_{\tau_{s}}^{2}\left(\frac{\tau_{s}}{\sigma_{\tau_{s}}}\right)^{2}\mathbb{E}_{Z\sim N(\frac{\tau_{s}}{\sigma_{\tau_{s}}},1)}\left[\left(\mathbf{1}\left\{ \frac{\tau_{s}}{\sigma_{\tau_{s}}}\geq0\right\} -\hat{\delta}_{1}(Z)\right)^{2}\right]\label{pf:lem2.2}
\end{align}
where the first equality follows from the definition, the second equality
follows from applying integration by-change-of-variable and letting
$z=\frac{x}{\sigma_{\tau_{s}}}$, and letting $\hat{\delta}_{1}(z)=\hat{\delta}(\sigma_{\tau_{s}}z)$.
As $\sigma_{\tau_{s}}^{2}$ is known, solving $\min_{\hat{\delta}\in\tilde{\mathcal{D}}}\sup_{\theta\in\Theta_{a_{e},a_{t}}}R_{sq}(\hat{\delta},\theta)$
is equivalent to solving 
\begin{equation}
\min_{\hat{\delta}_{1}}\sup_{\frac{\tau_{s}}{\sigma_{\tau_{s}}}}R_{sq}(\hat{\delta}_{1},\frac{\tau_{s}}{\sigma_{\tau_{s}}}),\label{pf:lem2.1}
\end{equation}
where $R_{sq}(\hat{\delta}_{1},\frac{\tau_{s}}{\sigma_{\tau_{s}}})=\left(\frac{\tau_{s}}{\sigma_{\tau_{s}}}\right)^{2}\mathbb{E}_{Z\sim N(\frac{\tau_{s}}{\sigma_{\tau_{s}}},1)}\left[\left(\mathbf{1}\left\{ \frac{\tau_{s}}{\sigma_{\tau_{s}}}\geq0\right\} -\hat{\delta}_{1}(Z)\right)^{2}\right]$
is the mean square regret of rule $\hat{\delta}_{1}$, a function
of $\frac{\hat{\tau}_{s}}{\sigma_{\tau_{s}}}\sim N(\frac{\tau_{s}}{\sigma_{\tau_{s}}},1)$
with an unknown mean $\frac{\tau_{s}}{\sigma_{\tau_{s}}}$ and unit
variance. As $\frac{\tau_{s}}{\sigma_{\tau_{s}}}=\frac{a_{t}s}{\left|a_{t}\right|\sigma}a_{e}\in[-\frac{a_{e}}{\sigma},\frac{a_{e}}{\sigma}]$,
by applying Lemma \ref{lem:main.1}, we find the solution of (\ref{pf:lem2.1})
as follows
\begin{align*}
\hat{\delta}_{1}^{*} & \left(\frac{\hat{\tau}_{s}}{\sigma_{\tau_{s}}}\right)=\begin{cases}
\frac{\exp\left(2\cdotp\tau^{*}\cdotp\frac{\hat{\tau}_{s}}{\sigma_{\tau_{s}}}\right)}{\exp\left(2\cdotp\tau^{*}\cdotp\frac{\hat{\tau}_{s}}{\sigma_{\tau_{s}}}\right)+1}, & \text{if }\frac{a_{e}}{\sigma}\geq\tau^{*},\\
\frac{\exp\left(2\cdotp\frac{a_{e}}{\sigma}\cdotp\frac{\hat{\tau}_{s}}{\sigma_{\tau_{s}}}\right)}{\exp\left(2\cdotp\frac{a_{e}}{\sigma}\cdotp\frac{\hat{\tau}_{s}}{\sigma_{\tau_{s}}}\right)+1}, & \text{if }\frac{a_{e}}{\sigma}<\tau^{*},
\end{cases}
\end{align*}
which coincides with $\hat{\delta}_{a_{e},a_{t}}^{*}$. Furthermore,
by applying Lemma \ref{lem:main.1} and (\ref{pf:lem2.2}), we derive
the worst-case mean square regret of $\hat{\delta}_{a_{e},a_{t}}^{*}$
as
\[
\sup_{\theta\in\Theta_{a_{e},a_{t}}}R_{sq}(\hat{\delta}_{a_{e},a_{t}}^{*},\theta)=\begin{cases}
\sigma_{\tau_{s}}^{2}\left(\tau^{*}\right)^{2}\rho(\tau^{*})=\left(\frac{a_{t}}{a_{e}}\right)^{2}\sigma^{2}\left(\tau^{*}\right)^{2}\rho(\tau^{*})\approx0.12\left(\frac{a_{t}}{a_{e}}\right)^{2}\sigma^{2}, & \frac{a_{e}}{\sigma}\geq\tau^{*},\\
\sigma_{\tau_{s}}^{2}\frac{a_{e}^{2}}{\sigma^{2}}\rho(\frac{a_{e}}{\sigma})=a_{t}^{2}\rho\left(\frac{a_{e}}{\sigma}\right), & \frac{a_{e}}{\sigma}<\tau^{*}.
\end{cases}
\]

\subsection*{Proof of Lemma \ref{lem:main.3}}

\subsubsection*{Proof of statement (i)}

When $\frac{a_{e}}{\sigma}\geq\tau^{*}$, 
\begin{align}
\sup_{\frac{a_{e}}{\sigma}\geq\tau^{*},a_{t}\in I(a_{e})}\sup_{\theta\in\Theta_{a_{e},a_{t}}}R_{sq}(\hat{\delta}_{a_{e},a_{t}}^{*},\theta) & =\sup_{\frac{a_{e}}{\sigma}\geq\tau^{*}}\left(\frac{a_{e}+k}{a_{e}}\right)^{2}\sigma^{2}\left(\tau^{*}\right)^{2}\rho(\tau^{*})\nonumber \\
 & =\left(1+\frac{k}{\tau^{*}\sigma}\right)^{2}\sigma^{2}\left(\tau^{*}\right)^{2}\rho(\tau^{*})\nonumber \\
 & =\sigma^{2}\left(\tau^{*}+\frac{k}{\sigma}\right)^{2}\rho(\tau^{*})\label{pf:lem3.1}
\end{align}
where the first equalify follows from $\theta_{t}\in[\theta_{e}-k,\theta_{e}+k]$,
and the second equality is because $\left(\frac{a_{e}+k}{a_{e}}\right)^{2}$
is decreasing in $a_{e}$. Similarly, when $0\leq\frac{a_{e}}{\sigma}<\tau^{*}$,
\begin{align}
\sup_{0\leq\frac{a_{e}}{\sigma}<\tau^{*},a_{t}\in I(a_{e})}\sup_{\theta\in\Theta_{a_{e},a_{t}}}R_{sq}(\hat{\delta}_{a_{e},a_{t}}^{*},\theta) & =\sup_{0\leq\frac{a_{e}}{\sigma}<\tau^{*}}\left(a_{e}+k\right)^{2}\rho\left(\frac{a_{e}}{\sigma}\right)\nonumber \\
 & =\sup_{0\leq\frac{a_{e}}{\sigma}<\tau^{*}}\sigma^{2}\left(\frac{a_{e}}{\sigma}+\frac{k}{\sigma}\right)^{2}\rho\left(\frac{a_{e}}{\sigma}\right)\nonumber \\
 & =\sigma^{2}\sup_{0\leq\tilde{a}_{e}<\tau^{*}}\left(\tilde{a}_{e}+\frac{k}{\sigma}\right)^{2}\rho\left(\tilde{a}_{e}\right).\label{pf:lem3.2}
\end{align}

Considering both (\ref{pf:lem3.1}) and (\ref{pf:lem3.2}), we see that finding  the worst-case one-dimensional subproblem is reduced to finding
\[
a^{*}\in\arg\sup_{0\leq\tilde{a}_{e}\leq\tau^*}(\tilde{a}_{e}+\frac{k}{\sigma})^{2}\rho\left(\tilde{a}_{e}\right).
\]
Since $\tilde{a}_{e}=\frac{a_{e}}{\sigma}$, the hardest one-dimensional subproblem corresponds to $a_{e}^*=\sigma a^*$, $a_{t}^*=\sigma a^*+k$. Applying Lemma \ref{lem:main.2} yields the formula for $\hat{\delta}^*_{\mathrm{H}}$ and the expression for $\sup_{\theta\in\Theta_{\mathrm{H}}}R_{sq}(\hat{\delta}_{\text{H}}^{*},\theta)$ as stated in (i) of the current lemma.

\subsubsection*{Proof of statement (ii)}

Write $g(\tilde{a}_{e}):=\left(\tilde{a}_{e}+\frac{k}{\sigma}\right)^{2}\rho\left(\tilde{a}_{e}\right)$,
which is a continuous and  differentiable function. Therefore, $a^{*}\in\arg\sup_{0\leq\tilde{a}_{e}\leq\tau^{*}}(\tilde{a}_{e}+\frac{k}{\sigma})^{2}\rho\left(\tilde{a}_{e}\right)$
is finite. First, we show $a^{*}>0$. Let $f^{(1)}(\cdot)$ be the
first derivative of function $f(\cdot)$. Algebra shows
\begin{align*}
\rho^{(1)}(\tilde{a}_{e}) & =\int2\left(\frac{1}{\exp\left(2\tilde{a}_{e}x\right)+1}\right)\left(-\frac{1}{\left(\exp\left(2\tilde{a}_{e}x\right)+1\right)^{2}}\right)\exp\left(2\tilde{a}_{e}x\right)2x\phi\left(x-\tilde{a}_{e}\right)dx\\
 & -\int\left(\frac{1}{\exp\left(2\tilde{a}_{e}x\right)+1}\right)^{2}\phi^{(1)}\left(x-\tilde{a}_{e}\right)dx\\
 & =-4\int\left(\frac{\exp\left(2\tilde{a}_{e}x\right)x}{\left(\exp\left(2\tilde{a}_{e}x\right)+1\right)^{3}}\phi\left(x-\tilde{a}_{e}\right)\right)dx\\
 & +\int\left(\frac{1}{\exp\left(2\tilde{a}_{e}x\right)+1}\right)^{2}(x-\tilde{a}_{e})\phi\left(x-\tilde{a}_{e}\right)dx.
\end{align*}
Thus, 
\[
\rho^{(1)}(0)=-\frac{1}{2}\int x\phi\left(x\right)dx+\frac{1}{4}\int x\phi\left(x\right)dx=0
\]
 as $\int x\phi\left(x\right)dx=0$. It follows then 
\[
g^{(1)}(\tilde{a}_{e})=2\left(\tilde{a}_{e}+\frac{k}{\sigma}\right)\rho\left(\tilde{a}_{e}\right)+\left(\tilde{a}_{e}+\frac{k}{\sigma}\right)^{2}\rho^{(1)}\left(\tilde{a}_{e}\right),
\]
and $g^{(1)}(0)=2\frac{k}{\sigma}\rho\left(0\right)=\frac{1}{2}\frac{k}{\sigma}>0$
as $k>0$. This implies that moving away from $\tilde{a}_{e}=0$ to
a small positive number always increases $g(\tilde{a}_{e})$. Thus,
$0$ is never a solution of $\sup_{0\leq\tilde{a}_{e}\leq1.23}g(\tilde{a}_{e})$. 

Next, we show $a^{*}<\tau^{*}$. By algebra, 
\begin{equation}
g^{(1)}(\tau^{*})=2\left(\tau^{*}+\frac{k}{\sigma}\right)\rho\left(\tau^{*}\right)+\left(\tau^{*}+\frac{k}{\sigma}\right)^{2}\rho^{(1)}\left(\tau^{*}\right).\label{pf:lem3.3}
\end{equation}
Note $\tau^{*}$ solves $\sup\limits _{\tau\in[0,\infty)}\tau^{2}\rho(\tau)$
and satisfiy the following FOC:
\begin{equation}
2\tau^{*}\rho(\tau^{*})+\left(\tau^{*}\right)^{2}\rho^{(1)}(\tau^{*})=0,\label{pf:lem3.4}
\end{equation}
implying 
\begin{equation}
\rho^{(1)}(\tau^{*})=-\frac{2\rho(\tau^{*})}{\tau^{*}}\label{pf:lem3.5}
\end{equation}
(\ref{pf:lem3.3}), (\ref{pf:lem3.4}) and (\ref{pf:lem3.5}) together
yield
\begin{align*}
g^{(1)}(\tau^{*}) & =2\frac{k}{\sigma}\rho\left(\tau^{*}\right)+\left(\frac{k^{2}}{\sigma^{2}}+2\tau^{*}\frac{k}{\sigma}\right)\rho^{(1)}\left(\tau^{*}\right)\\
 & =-2\rho\left(\tau^{*}\right)\left[\frac{k}{\sigma}+\frac{k^{2}}{\tau^{*}\sigma^{2}}\right]<0,
\end{align*}
 implying $\tau^{*}$ is not a solution of $\sup_{0\leq\tilde{a}_{e}\leq1.23}g(\tilde{a}_{e})$.
\subsubsection*{Proof of statement (iii)}
By statement (ii), $a^*$ is an interior solution and must satisfy the following FOC:
\[
2\left(a^*+\frac{k}{\sigma}\right)\rho\left(a^*\right)+\left(a^*+\frac{k}{\sigma}\right)^{2}\rho^{(1)}\left(a^*\right)=0.
\]

As $(a^*+\frac{k}{\sigma})>0$, $a^*$ must also satisfy
\begin{equation}\label{pf:FOC.a.star}
2\rho\left(a^*\right)+\left(a^*+\frac{k}{\sigma}\right)\rho^{(1)}\left(a^*\right)=0.  
\end{equation}
Moreover, as $a^*$ is a local maximum of a continuously differentiable function, it must also satisfy the following second-order condition:
\begin{equation}\label{pf:SOC.a.star}
3\rho^{(1)}(a^{*})+(a^{*}+\frac{k}{\sigma})\rho^{(2)}(a^{*})<0.   
\end{equation}
Viewing the right-hand-side of (\ref{pf:FOC.a.star}) as a function of $a^*$ and $k$, say $F(a^*,k)$, we may write
\[
\frac{\partial a^{*}}{\partial k}=-\frac{\frac{\partial F(a^{*},k)}{\partial k}}{\frac{\partial F(a^{*},k)}{\partial a^{*}}}=-\frac{\frac{1}{\sigma}\rho^{(1)}(a^{*})}{3\rho^{(1)}(a^{*})+(a^{*}+\frac{k}{\sigma})\rho^{(2)}(a^{*})}.
\]
From (\ref{pf:FOC.a.star}), we know $\rho^{(1)}(a^*)<0$. Together with (\ref{pf:SOC.a.star}), we conclude that $\frac{\partial a^{*}}{\partial k}<0$. The proof for $\frac{\partial a^{*}}{\partial \sigma}>0$
is similar and omitted.

\subsection*{Proof of Theorem \ref{thm:main.1}}

Firstly, note the following inequalities hold:
\begin{align}
\sup_{\theta\in\Theta}R_{sq}(\hat{\delta}_{\text{H}}^{*},\theta) & \geq\sup_{\theta\in\Theta}R_{sq}(\hat{\delta}^{*},\theta)\nonumber \\
 & \geq\sup_{\theta\in\Theta_{\mathrm{H}}}R_{sq}(\hat{\delta}^{*},\theta)\nonumber \\
 & \geq\sup_{\theta\in\Theta_{\mathrm{H}}}R_{sq}(\hat{\delta}_{\mathrm{H}}^{*},\theta),\label{pf:thm1.1}
\end{align}
where the first inequality follows from the definition of $\hat{\delta}^{*}$,
the second relation follows from $\Theta_{\mathrm{H}}\subseteq\Theta$,
and the third relation follows from the fact that $\hat{\delta}_{\mathrm{H}}^{*}$
is a minimax optimal rule of the hardest one-dimensional subproblem.
Secondly, Theorem \ref{thm:2} establishes 
\begin{equation}
\sup_{\theta\in\Theta}R_{sq}(\hat{\delta}_{\text{H}}^{*},\theta)\leq\sup_{\theta\in\Theta_{\mathrm{H}}}R_{sq}(\hat{\delta}_{\mathrm{H}}^{*},\theta).\label{pf:thm1.2}
\end{equation}
Combining (\ref{pf:thm1.1}) and (\ref{pf:thm1.2}) yields the desired
conclusion. 

\section{Additional technical results}

Recall the definition of $a^*$ in (\ref{eq:a.star}). Let $\rho^{*}\left(\tilde{\theta}_{e}\right)=\int\left(\frac{1}{\exp\left(2\cdotp a^{*}\cdotp y\right)+1}\right)^{2}\phi(y-\tilde{\theta}_{e})dy$. 
\begin{thm}
$\sup_{\theta\in\Theta}R_{sq}(\hat{\delta}_{\mathrm{H}}^{*},\theta)\leq\sup_{\theta\in\Theta_{\mathrm{H}}}R_{sq}(\hat{\delta}_{\mathrm{H}}^{*},\theta)$.\label{thm:2}
\end{thm}
\begin{proof}
By Lemma \ref{lem:4}, $\sup_{\theta\in\Theta}R_{sq}(\hat{\delta}_{\text{H}}^{*},\theta)=\sigma^{2}\sup_{-\frac{k}{\sigma}\leq\tilde{a}_{e}<\infty}\left(\tilde{a}_{e}+\frac{k}{\sigma}\right)^{2}\rho^{*}\left(\tilde{a}_{e}\right)$.
By Lemma \ref{lem:main.3}, $\hat{\delta}_{\mathrm{H}}^{*}$ is a
minimax rule with respect to the hardest one-dimensional problem,
and it holds 
\begin{align*}
\sup_{\theta\in\Theta_{\mathrm{H}}}R_{sq}(\hat{\delta}_{\mathrm{H}}^{*},\theta) & =\sigma^{2}\sup_{0\leq\tilde{a}_{e}\leq1.23}(\tilde{a}_{e}+\frac{k}{\sigma})^{2}\rho\left(\tilde{a}_{e}\right)=\sigma^{2}(a^{*}+\frac{k}{\sigma})^{2}\rho\left(a^{*}\right).
\end{align*}
Furthermore, Lemma \ref{lem:5} establishes 
\[
\sup_{-\frac{k}{\sigma}\leq\tilde{a}_{e}<\infty}\left(\tilde{a}_{e}+\frac{k}{\sigma}\right)^{2}\rho^{*}\left(\tilde{a}_{e}\right)\leq({a}^{*}+\frac{k}{\sigma})^{2}\rho\left({a}^{*}\right),
\]
yielding the conclusion.
\end{proof}
\begin{lem}
$\sup_{\theta\in\Theta}R_{sq}(\hat{\delta}_{\text{H}}^{*},\theta)=\sigma^{2}\sup_{-\frac{k}{\sigma}\leq\tilde{a}_{e}<\infty}\left(\tilde{a}_{e}+\frac{k}{\sigma}\right)^{2}\rho^{*}\left(\tilde{a}_{e}\right)$.\label{lem:4}
\end{lem}
\begin{proof}
For any $\left(\begin{array}{c}
\theta_{e}\\
\theta_{t}
\end{array}\right)\in\Theta,$ note $\left(\begin{array}{c}
-\theta_{e}\\
-\theta_{t}
\end{array}\right)\in\Theta.$ Thus, consider each $\theta=\left(\begin{array}{c}
\theta_{e}\\
\theta_{t}
\end{array}\right)\in\Theta$ where $\theta_{t}\geq0$. Applying change-of-variable yields
\begin{align*}
R_{sq}(\hat{\delta}_{\text{H}}^{*},-\theta)= & \theta_{t}^{2}\mathbb{E}_{-\theta_{e}}\left[\left(\mathbf{1}\left\{ -\theta_{t}\geq0\right\} -\frac{\exp\left(2\cdotp a^{*}\cdotp\frac{\hat{\theta}_{e}}{\sigma}\right)}{\exp\left(2\cdotp a^{*}\cdotp\frac{\hat{\theta}_{e}}{\sigma}\right)+1}\right)^{2}\right]\\
= & \theta_{t}^{2}\int\left(\frac{\exp\left(2\cdotp\frac{a^{*}y}{\sigma}\right)}{\exp\left(2\cdotp\frac{a^{*}y}{\sigma}\right)+1}\right)^{2}\frac{\phi(\frac{y+\theta_{e}}{\sigma})}{\sigma}dy\\
= & \theta_{t}^{2}\int\left(\frac{\exp\left(-2\cdotp\frac{a^{*}y}{\sigma}\right)}{\exp\left(-2\cdotp\frac{a^{*}y}{\sigma}\right)+1}\right)^{2}\frac{\phi\left(\frac{-y+\theta_{e}}{\sigma}\right)}{\sigma}dy\\
= & \theta_{t}^{2}\int\left(\frac{1}{1+\exp\left(2\cdotp\frac{a^{*}y}{\sigma}\right)}\right)^{2}\frac{\phi\left(\frac{y-\theta_{e}}{\sigma}\right)}{\sigma}dy\\
= & \theta_{t}^{2}\mathbb{E}_{\theta_{e}}\left[\left(\frac{1}{1+\exp\left(2\cdotp\frac{a^{*}\hat{\theta}_{e}}{\sigma}\right)}\right)^{2}\right]\\
= & R_{sq}(\hat{\delta}_{\text{H}}^{*},\theta).
\end{align*}
Therefore, we deduce
\begin{align*}
\sup_{\theta\in\Theta}R_{sq}(\hat{\delta}_{\text{H}}^{*},\theta) & =\sup_{\theta_{e}\in\mathbb{R},\theta_{t}\in I(\theta_{e})}\theta_{t}^{2}\mathbb{E}_{\theta_{e}}\left[\left(\mathbf{1}\{\theta_{t}\geq0\}-\frac{\exp\left(2\cdotp a^{*}\cdotp\frac{\hat{\theta}_{e}}{\sigma}\right)}{\exp\left(2\cdotp a^{*}\cdotp\frac{\hat{\theta}_{e}}{\sigma}\right)+1}\right)^{2}\right],\\
 & =\max\left\{ \sup_{\theta\in\Theta,\theta_{t}\geq0}R_{sq}(\hat{\delta}_{\text{H}}^{*},\theta),\sup_{\theta\in\Theta,\theta_{t}<0}R_{sq}(\hat{\delta}_{\text{H}}^{*},\theta)\right\} ,\\
 & =\sup_{\theta\in\Theta,\theta_{t}\geq0}R_{sq}(\hat{\delta}_{\text{H}}^{*},\theta)\\
 & =\sup_{\theta\in\Theta,\theta_{t}\geq0}\theta_{t}^{2}\mathbb{E}_{\theta_{e}}\left[\left(\frac{1}{\exp\left(2\cdotp a^{*}\cdotp\frac{\hat{\theta}_{e}}{\sigma}\right)+1}\right)^{2}\right].
\end{align*}
Moreover, note
\begin{align*}
 & \sup_{\theta\in\Theta,\theta_{t}\geq0}\theta_{t}^{2}\mathbb{E}_{\theta_{e}}\left[\left(\frac{1}{\exp\left(2\cdotp a^{*}\cdotp\frac{\hat{\theta}_{e}}{\sigma}\right)+1}\right)^{2}\right]\\
 & =\sup_{-k\leq\theta_{e}<\infty}(\theta_{e}+k)^{2}\mathbb{E}_{\theta_{e}}\left[\left(\frac{1}{\exp\left(2\cdotp a^{*}\cdotp\frac{\hat{\theta}_{e}}{\sigma}\right)+1}\right)^{2}\right]\\
 & =\sigma^{2}\sup_{-\frac{k}{\sigma}\leq\frac{\theta_{e}}{\sigma}<\infty}\left(\frac{\theta_{e}}{\sigma}+\frac{k}{\sigma}\right)^{2}\rho^{*}\left(\frac{\theta_{e}}{\sigma}\right)\\
 & =\sigma^{2}\sup_{-\frac{k}{\sigma}\leq\tilde{\theta}_{e}<\infty}\left(\tilde{\theta}_{e}+\frac{k}{\sigma}\right)^{2}\rho^{*}\left(\tilde{\theta}_{e}\right),
\end{align*}
where $\rho^{*}\left(\tilde{\theta}_{e}\right)=\int\left(\frac{1}{\exp\left(2\cdotp a^{*}\cdotp y\right)+1}\right)^{2}\phi(y-\tilde{\theta}_{e})dy$. 
\end{proof}
\begin{lem}
$\sup_{-\frac{k}{\sigma}\leq\tilde{\theta}_{e}<\infty}\left(\tilde{\theta}_{e}+\frac{k}{\sigma}\right)^{2}\rho^{*}\left(\tilde{\theta}_{e}\right)\leq(a^{*}+\frac{k}{\sigma})^{2}\rho\left(a^{*}\right)$.\label{lem:5}
\end{lem}
\begin{proof}
Recall $g(\tilde{a}_{e}):=\left(\tilde{a}_{e}+\frac{k}{\sigma}\right)^{2}\rho\left(\tilde{a}_{e}\right)$. Write $g^{*}(\tilde{\theta}_{e}):=\left(\tilde{\theta}_{e}+\frac{k}{\sigma}\right)^{2}\rho^{*}\left(\tilde{\theta}_{e}\right)$.
Note $g(a^{*})=g^{*}(a^{*})$ as $\rho^{*}\left(a^{*}\right)=\rho\left(a^{*}\right)$.
Thus, it suffices to show that $a^{*}$ solves $\sup_{-\frac{k}{\sigma}\leq\tilde{\theta}_{e}<\infty}g^{*}(\tilde{\theta}_{e})$.
We take two steps:

Step 1: we show that $a^{*}$ is a local extremum point of $g^{*}(\tilde{\theta}_{e})$.
To see this, note $a^{*}\in\arg\sup_{0\leq\tilde{a}_{e}\leq1.23}(\tilde{a}_{e}+\frac{k}{\sigma})^{2}\rho\left(\tilde{a}_{e}\right).$
By Lemma \ref{lem:main.3} (ii), $a^{*}$ is an interior point in
$[0,\tau^{*}]$. Therefore, $a^{*}$ must satisfy the following FOC
\begin{align*}
2\left(a^{*}+\frac{k}{\sigma}\right)\rho\left(a^{*}\right)+\left(a^{*}+\frac{k}{\sigma}\right)^{2}\rho^{(1)}\left(a^{*}\right) & =0.
\end{align*}
As $a^{*}+\frac{k}{\sigma}>0$, it implies 
\begin{equation}
2\rho\left(a^{*}\right)+\left(a^{*}+\frac{k}{\sigma}\right)\rho^{(1)}\left(a^{*}\right)=0.\label{pf:lem5.1}
\end{equation}

We evaluate the first derivate of $g^{*}(\cdotp)$ at $a^{*}$:
\begin{align}
\left(g^{*}\right)^{(1)}(a^{*}) & =\left(a^{*}+\frac{k}{\sigma}\right)\left[2\rho^{*}\left(a^{*}\right)+\left(a^{*}+\frac{k}{\sigma}\right)\rho^{*}{}^{(1)}\left(a^{*}\right)\right]\nonumber \\
 & =\left(a^{*}+\frac{k}{\sigma}\right)\left[2\rho\left(a^{*}\right)+\left(a^{*}+\frac{k}{\sigma}\right)\left(\rho\right)^{(1)}\left(a^{*}\right)\right]\nonumber \\
 & =0,\label{pf:lem5.2}
\end{align}
where the second equality follows from Lemma \ref{lem:6}, and from
using $\rho\left(a^{*}\right)=\rho^{*}\left(a^{*}\right)$ again,
and the third equality follows from (\ref{pf:lem5.1}). Thus, we conclude
that $a^{*}$ is also a local extremum point of $g^{*}(\cdotp)$

Step 2: we show $a^{*}$ is in fact a global maximum of the problem
$\sup_{-\frac{k}{\sigma}\leq\tilde{\theta}_{e}<\infty}g^{*}(\tilde{\theta}_{e})$.
We analyze $\left(g^{*}\right)^{(1)}(\tilde{\theta}_{e})$ more in
detail. Algebra shows
\begin{align*}
\left(g^{*}\right)^{(1)}(\tilde{\theta}_{e}) & =\left(\tilde{\theta}_{e}+\frac{k}{\sigma}\right)\mathbf{g}(\tilde{\theta}_{e}),
\end{align*}
where $\mathbf{g}(\tilde{\theta}_{e})=2\rho^{*}\left(\tilde{\theta}_{e}\right)+\left(\tilde{\theta}_{e}+\frac{k}{\sigma}\right)\left(\rho^{*}\right)^{(1)}\left(\tilde{\theta}_{e}\right)$.
As $\tilde{\theta}_{e}+\frac{k}{\sigma}\geq0$, it follows the sign
of $\left(g^{*}\right)^{(1)}(\tilde{\theta}_{e})$ only depends on
$\mathbf{g}(\tilde{\theta}_{e})$, which we further analyze below.
To this end, write $\frac{1}{1+\exp\left(2\cdotp a^{*}\cdotp y\right)}:=w^{*}(y)$.
Using integration by parts twice, it follows 
\begin{align*}
 & \mathbf{g}(\tilde{\theta}_{e})\\
= & 2\int w^{*}(y)^{2}\phi(y-\tilde{\theta}_{e})dy+(\tilde{\theta}_{e}+\frac{k}{\sigma})\int w^{*}(y)^{2}\frac{d\phi(y-\tilde{\theta}_{e})}{d\tilde{\theta}_{e}}dy\\
= & 2\int w^{*}(y)^{2}\phi(y-\tilde{\theta}_{e})dy-(\tilde{\theta}_{e}+\frac{k}{\sigma})\int w^{*}(y)^{2}\frac{d\phi(y-\tilde{\theta}_{e})}{dy}dy\\
= & 2\int w^{*}(y)^{2}\phi(y-\tilde{\theta}_{e})dy-(\theta_{e}+\frac{k}{\sigma})\int w^{*}(y)^{2}d\phi(y-\tilde{\theta}_{e})\\
= & 2\left(\int w^{*}(y)^{2}\phi(y-\tilde{\theta}_{e})dy+\int w^{*}(y)\frac{dw^{*}(y)}{dy}(\theta_{e}+\frac{k}{\sigma})\phi(y-\tilde{\theta}_{e})dy\right)\\
= & 2\left(\int w^{*}(y)^{2}\phi(y-\tilde{\theta}_{e})dy+\int w^{*}(y)\frac{dw^{*}(y)}{dy}(\theta_{e}-y)\phi(y-\tilde{\theta}_{e})dy+\int w^{*}(y)\frac{dw^{*}(y)}{dy}(\frac{k}{\sigma}+y)\phi(y-\tilde{\theta}_{e})dy\right)\\
= & 2\left(\int w^{*}(y)^{2}\phi(y-\tilde{\theta}_{e})dy+\int\phi(y-\tilde{\theta}_{e})w^{*}(y)\frac{dw^{*}(y)}{dy}d\phi(y-\tilde{\theta}_{e})+\int w^{*}(y)\frac{dw^{*}(y)}{dy}(\frac{k}{\sigma}+y)\phi(y-\tilde{\theta}_{e})dy\right)\\
= & 2\left(\int w^{*}(y)^{2}\phi(y-\tilde{\theta}_{e})dy-\int\frac{d\left(\frac{\partial w^{*}(y)}{\partial y}w^{*}(y)\right)}{dy}\phi(y-\tilde{\theta}_{e})dy+\int w^{*}(y)\frac{dw^{*}(y)}{dy}(\frac{k}{\sigma}+y)\phi(y-\tilde{\theta}_{e})dy\right)\\
= & 2\int\mathbf{w}(y)\phi(y-\tilde{\theta}_{e})dy,
\end{align*}
where 
\begin{equation}
\mathbf{w}(y)=w^{*}(y)^{2}-\left(\frac{dw^{*}(y)}{dy}\right)^{2}-\frac{d^{2}w^{*}(y)}{dy^{2}}w^{*}(y)+w^{*}(y)\frac{dw^{*}(y)}{dy}(\frac{k}{\sigma}+y).\label{pf:bold.w}
\end{equation}

Lemma \ref{lem:7} shows that $\int\mathbf{w}(y)\phi(y-\tilde{\theta}_{e})dy$
has a unique sign change from $+$ to $-$ at $a^*$, which verifies immediately
that $a^{*}$ is in fact a global maximum of the problem $\sup_{-\frac{k}{\sigma}\leq\tilde{\theta}_{e}<\infty}g^{*}(\tilde{\theta}_{e})$. 
\end{proof}
\begin{lem}
$\rho^{(1)}(a^{*})=\left(\rho^{*}\right)^{(1)}(a^{*})$ . \label{lem:6}
\end{lem}
\begin{proof}
Note for all $\tilde{\theta}_{e}\in\mathbb{R}$:

\begin{align*}
\left(\rho^{*}\right)^{(1)}\left(\tilde{\theta}_{e}\right) & =-\int\left(\frac{1}{\exp\left(2\cdotp a^{*}\cdotp y\right)+1}\right)^{2}\phi^{(1)}(y-\tilde{\theta}_{e})dy,
\end{align*}
while algebra shows 
\begin{align*}
\rho^{(1)}(\tilde{\theta}_{e}) & =F_{1}(\tilde{\theta}_{e})-\int\left(\frac{1}{\exp\left(2\tilde{\theta}_{e}y\right)+1}\right)^{2}\phi^{(1)}\left(x-\tilde{\theta}_{e}\right)dy,
\end{align*}
where $F_{1}(\tilde{\theta}_{e})=-4\int\frac{\exp\left(2\tilde{\theta}_{e}y\right)\phi\left(y-\tilde{\theta}_{e}\right)}{\left(\exp\left(2\tilde{\theta}_{e}y\right)+1\right)^{3}}ydy$.
We can further verify that $F_{1}(\tilde{\theta}_{e})=-4\int w_{\tilde{\theta}_{e}}(y)ydy$,
where 
\[
w_{\tilde{\theta}_{e}}(y)=\frac{\phi^{2}\left(y-\tilde{\theta}_{e}\right)\phi^{2}\left(y+\tilde{\theta}_{e}\right)}{\left(\phi\left(y-\tilde{\theta}_{e}\right)+\phi\left(y+\tilde{\theta}_{e}\right)\right)^{3}}
\]
is such that $w_{\tilde{\theta}_{e}}(y)=w_{\tilde{\theta}_{e}}(-y)$
for all $y$. Thus, it holds $F_{1}(\tilde{\theta}_{e})=0$ for all
$\tilde{\theta}_{e}\in\mathbb{R}$. It then holds 
\begin{align*}
\rho^{(1)}(\tilde{\theta}_{e}) & =-\int\left(\frac{1}{\exp\left(2\tilde{\theta}_{e}y\right)+1}\right)^{2}\phi^{(1)}\left(x-\tilde{\theta}_{e}\right)dy.
\end{align*}

Evaluating $\left(\rho^{*}\right)^{(1)}\left(\tilde{\theta}_{e}\right)$
and $\rho^{(1)}(\tilde{\theta}_{e})$ at $a^{*}$ yields the conclusion.
\end{proof}
\begin{lem}
$\mathbf{g}(\tilde{\theta}_{e})$ has a unique sign change from $+$
to $-$ at $a^{*}$.\label{lem:7}
\end{lem}
\begin{proof}Note by Lemma \ref{lem:5}, $\mathbf{g}(\tilde{\theta}_{e})=2\int\mathbf{w}(y)\phi(y-\tilde{\theta}_{e})dy$,
where $\mathbf{w}(y)$ is defined in (\ref{pf:bold.w}). Also,
\begin{align*}
w^{*}(y) & =\frac{1}{1+\exp\left(2\cdotp a^{*}\cdotp y\right)},\\
\frac{dw^{*}(y)}{dy} & =-\left(w^{*}(y)\right)^{2}\exp\left(2\cdotp a^{*}\cdotp y\right)2a^{*},\\
\frac{d^{2}w^{*}(y)}{dy^{2}} & =2\left(w^{*}(y)\right)^{3}\left(\exp\left(2\cdotp a^{*}\cdotp y\right)2a^{*}\right)^{2}-\left(w^{*}(y)\right)^{2}\exp\left(2\cdotp a^{*}\cdotp y\right)\left(2a^{*}\right)^{2},\\
\hat{\delta}_{\mathrm{H}}^{*}(y) & =w^{*}(y)\exp\left(2\cdotp a^{*}\cdotp y\right).
\end{align*}
Thus, 
\begin{align*}
\mathbf{w}(y) & =w^{*}(y)^{2}-3\left(w^{*}(y)\right)^{4}\left(\exp\left(2\cdotp a^{*}\cdotp y\right)2a^{*}\right)^{2}\\
 & +\left(w^{*}(y)\right)^{3}\exp\left(2\cdotp a^{*}\cdotp y\right)\left(2a^{*}\right)^{2}\\
 & -\left(w^{*}(y)\right)^{3}\exp\left(2\cdotp a^{*}\cdotp y\right)2a^{*}(\frac{k}{\sigma}+y)\\
 & =w^{*}(y)^{2}\hat{\delta}_{\mathrm{H}}^{*}(y)\left\{ \frac{1}{\hat{\delta}_{\mathrm{H}}^{*}(y)}-3\hat{\delta}_{\mathrm{H}}^{*}(y)\left(2a^{*}\right)^{2}+\left(2a^{*}\right)^{2}-2a^{*}(\frac{k}{\sigma}+y)\right\} \\
 & =w^{*}(y)^{2}\hat{\delta}_{\mathrm{H}}^{*}(y)\tilde{\mathbf{w}}(y)
\end{align*}
where 
\begin{align}
\tilde{\mathbf{w}}(y) & =\frac{1}{\hat{\delta}_{\mathrm{H}}^{*}(y)}-3\hat{\delta}_{\mathrm{H}}^{*}(y)\left(2a^{*}\right)^{2}+\left(2a^{*}\right)^{2}-2a^{*}\left(\frac{k}{\sigma}+y\right).\label{pf:lem7.1}
\end{align}
As $w^{*}(y)^{2}\hat{\delta}_{\mathrm{H}}^{*}(y)>0$, the sign of
$\mathbf{w}(y_{1})$ is determined by $\tilde{\mathbf{w}}(y_{1})$.
It is straightforward to verify that 
\[
\frac{d\tilde{\mathbf{w}}(y)}{dy}=-\left(\frac{1}{\left(\hat{\delta}_{\mathrm{H}}^{*}(y)\right)^{2}}+12\left(a^{*}\right)^{2}\right)\frac{d\hat{\delta}_{\mathrm{H}}^{*}(y)}{dy}-2a^{*}<0.
\]
Thus, it holds that $\tilde{\mathbf{w}}(y)$ is strictly decreasing
and has at most one sign change from $+$ to $-$. Moreover, note
$\lim_{y\rightarrow-\infty}\tilde{\mathbf{w}}(y)=\infty$, and $\lim_{y\rightarrow\infty}\tilde{\mathbf{w}}(y)=-\infty$.
Thus, $\tilde{\mathbf{w}}(y)$ has one and only one sign change from
$+$ to $-$, implying that $\mathbf{w}(y)$ has one and only one
sign change from $+$ to $-$ as well. It follows from \citet[Theorem C.1(i), ][]{kitagawa2022treatment} that $\mathbf{g}(\tilde{\theta}_{e})$ at
most has one sign change. 

Next, we show that $\mathbf{g}(\tilde{\theta}_{e})$ indeed has one
sign change at $a^{*}$. To this end, note
\begin{align*}
\mathbf{g}(a^{*}) & =2\int\mathbf{w}(y)\phi(y-a^{*})dy=0\\
\mathbf{g}^{(1)}(a^{*}) & =2\int\mathbf{w}(y)\left(y-a^{*}\right)\phi(y-a^{*})dy\\
 & =2\int\mathbf{w}(y)y\phi(y-a^{*})dy-\underset{0}{\underbrace{2a^{*}\int\mathbf{w}(y)\phi(y-a^{*})dy}}\\
 & =2\int\mathbf{w}(y)y\phi(y-a^{*})dy.
\end{align*}
Algebra shows
\begin{align*}
\mathbf{w}(y) & =w^{*}(y)^{2}\hat{\delta}_{\mathrm{H}}^{*}(y)\tilde{\mathbf{w}}(y)\\
 & =\left(1-\hat{\delta}_{\mathrm{H}}^{*}(y)\right)^{2}\hat{\delta}_{\mathrm{H}}^{*}(y)\tilde{\mathbf{w}}(y)\\
 & =\left(\frac{\phi\left(y+a^{*}\right)}{\phi\left(y-a^{*}\right)+\phi\left(y+a^{*}\right)}\right)^{2}\frac{\phi\left(y-a^{*}\right)}{\phi\left(y-a^{*}\right)+\phi\left(y+a^{*}\right)}\tilde{\mathbf{w}}(y),
\end{align*}
Thus,
\begin{align*}
\mathbf{g}^{(1)}(a^{*})= & 2\int w_{a^{*}}(y)\tilde{\mathbf{w}}(y)ydy,
\end{align*}
where $w_{a^{*}}(y)=\frac{\phi^{2}\left(y-a^{*}\right)\phi^{2}\left(y+a^{*}\right)}{\left(\phi\left(y-a^{*}\right)+\phi\left(y+a^{*}\right)\right)^{3}}>0$
and is such that $w_{a^{*}}(-y)=w_{a^{*}}(y)$ for all $y\in\mathbb{R}$,
and $\tilde{\mathbf{w}}(y)$ is strictly decreasing from $+\infty$
to $-\infty$. Let $t^{*}$ be the unique point such that $\tilde{\mathbf{w}}(t^{*})=0$.
Suppose $t^{*}\geq0$. Then, we have the following decomposition
\begin{align*}
\mathbf{g}^{(1)}(a^{*}) & =2\int_{y<-t^{*}}w_{a^{*}}(y)\tilde{\mathbf{w}}(y)ydy\\
 & +2\int_{-t^{*}\leq y<t^{*}}w_{a^{*}}(y)\tilde{\mathbf{w}}(y)ydy\\
 & +2\int_{y>t^{*}}w_{a^{*}}(y)\tilde{\mathbf{w}}(y)ydy,
\end{align*}
where all three terms above can be signed to be negative. A similar
decomposition also reveals that $\mathbf{g}^{(1)}(a^{*})<0$ holds
true when $t^{*}<0$. Thus, we we conclude that $\mathbf{g}^{(1)}(a^{*})<0$
and $a^{*}$ is indeed a sign change of $\mathbf{g}$. Then, we apply
\citet[Theorem C.1(i), ][]{kitagawa2022treatment} to conclude that $\mathbf{g}(\tilde{\theta}_{e})$
indeed has one and only on sign change at $a^{*}$. Furthermore,  \citet[Theorem C.1(ii), ][]{kitagawa2022treatment} implies that $\mathbf{g}(\tilde{\theta}_{e})$
and $\mathbf{w}(y)$ in the same order. The conclusion follows.
\end{proof}
\begin{lem}
\label{lem:8}Let $0<c<\tau^{*}$. Then, it holds $\sup_{\tau\in[-c,c]}R_{sq}(\hat{\delta}_{\pi_{c}},\tau)=c^{2}\rho(c)$.
\end{lem}
\begin{proof}
By a symmetry argument, it can be shown that $R_{sq}(\hat{\delta}_{\pi_{c}},\tau)=R_{sq}(\hat{\delta}_{\pi_{c}},-\tau)$
for all $\tau$. Thus, 
\[
\sup_{\tau\in[-c,c]}R_{sq}(\hat{\delta}_{\pi_{c}},\tau)=\sup_{\tau\in[0,c]}g_{c}^{*}(\tau),
\]
where we define $g_{c}^{*}(\tau):=\tau^{2}\rho_{c}^{*}(\tau)$, and $\rho_{c}^{*}(\tau):=\int\left(\frac{1}{\exp\left(2\cdotp c\cdotp y\right)+1}\right)^{2}\phi\left(y-\tau\right)dy$.
As $g_{c}^{*}(c)=c^{2}\rho(c)$, it suffices to show that 
\[
c\in\arg\sup_{\tau\in[0,c]}g_{c}^{*}(\tau).
\]
Below we show that $g_{c}^{*}(\cdotp)$ is increasing in $[0,c]$,
and the conclusion will follow. We take two steps.

Step 1: show $\left(g_{c}^{*}\right)(\cdot)$ is first increasing and
then decreasing in $[0,\infty)$. Note $\left(g_{c}^{*}\right)(\cdot)$ may be analyzed by using the same technique employed in \citet[Lemma C.5,][]{kitagawa2022treatment}. That is, by  first re-writing $\left(g_{c}^{*}\right)^{(1)}(\cdot)$ using change-of-variable twice, and then invoking \citet[Theorem C.1, ][]{kitagawa2022treatment}, we can conclude that $\left(g_{c}^{*}\right)^{(1)}(\cdot)$
at most has one sign change in $[0,\infty)$. Furthermore, note $g_{c}^{*}(0)=0$, $\lim_{\tau\rightarrow \infty}g_{c}^{*}(\tau)=0$, and $g_{c}^{*}(\tau)>0$ at any $0<\tau<\infty$. As $g_{c}^{*}$ is a continuous and differentiable function, there must exist some $0<x<\infty$ such that  $g_{c}^{*}(x)\geq g_{c}^{*}(\tau)$ for all $\tau\in[0,\infty)$ with the inequality strict for some  $\tau\in[0,x)$ and $\tau\in(x,\infty)$. Thus, $\left(g_{c}^{*}\right)^{(1)}(\cdot)$
at least has one sign change in $[0,\infty)$. Applying \citet[Theorem C.1, ][]{kitagawa2022treatment}, we conclude that $\left(g_{c}^{*}\right)^{(1)}(\cdot)$
has a unique sign change from + to $-$ in $[0,\infty)$, implying
that $\left(g_{c}^{*}\right)(\tau)$ is first increasing and then
decreasing in $[0,\infty)$.

Step 2: show $\left(g_{c}^{*}\right)^{(1)}(c)\geq0$. Suppose not.
Then, by the conclusion from the first step, it must hold that $\left(g_{c}^{*}\right)^{(1)}(c)<0$
and 
\[
\left(g_{c}^{*}\right)(c)>\left(g_{c}^{*}\right)(\tau^{*}),
\]
as $c<\tau^{*}$. Furthermore, by Lemma \ref{lem:9}, we know $\left(g_{c}^{*}\right)(\tau^{*})>\left(g_{\tau^{*}}^{*}\right)(\tau^{*})$,
and $\left(g_{c}^{*}\right)(c)<\left(g_{\tau^{*}}^{*}\right)(c)$
for all $0<c<\tau^{*}$, implying 
\begin{equation}
\left(g_{\tau^{*}}^{*}\right)(c)>\left(g_{\tau^{*}}^{*}\right)(\tau^{*}).\label{pf:lem.8.1}
\end{equation}
However, we know it must hold that $\left(g_{\tau^{*}}^{*}\right)(\tau^{*})>\left(g_{\tau^{*}}^{*}\right)(c)$
as $\left(g_{\tau^{*}}^{*}\right)(\tau^{*})$ corresponds to the worst-case mean square regret of the global minimax optimal rule. Therefore, it must hold that $\left(g_{c}^{*}\right)^{(1)}(c)\geq0$.
And we conclude that $g_{c}^{*}(\cdotp)$ is increasing in $[0,c]$
by combining steps 1 and 2.
\end{proof}
\begin{lem}\label{lem:9} \quad

\begin{itemize}
\item[(i)] $\left(g_{c}^{*}\right)(\tau^{*})>\left(g_{\tau^{*}}^{*}\right)(\tau^{*})$ for all $0<c<\tau^{*}$;

\item[(ii)] $\left(g_{c}^{*}\right)(c)<\left(g_{\tau^{*}}^{*}\right)(c)$
for all $0<c<\tau^{*}$.
\end{itemize}
\end{lem}
\begin{proof}
Recall the definition of $g^*_{a}(b)$:
\[
g_{a}^{*}(b):=b^{2}\int\left(\frac{1}{\exp\left(2\cdotp a\cdotp y\right)+1}\right)^{2}\phi\left(y-b\right)dy.\]
Statement (i). Viewing $g_{c}^{*}(\tau^{*})$ as a function of $c$,
we aim to establish that $\frac{\partial\left(g_{c}^{*}\right)(\tau^{*})}{\partial c}<0$
for all $0<c<\tau^{*}$, and statement (i) will follow directly. For
all $0<c<\tau^{*}$:
\begin{align*}
\frac{\partial\left(g_{c}^{*}\right)(\tau^{*})}{\partial c} & =-4\left(\tau^{*}\right)^{2}\int\frac{\exp(2cy)y}{\left(\exp(2cy)+1\right)^{3}}\phi(y-\tau^{*})dy\\
 & =-4\left(\tau^{*}\right)^{2}\int\frac{\phi^{2}(y+c)\phi(y-c)}{\left(\phi(y+c)+\phi(y-c)\right)^{3}}y\phi(y-\tau^{*})dy.
\end{align*}
To show $\frac{\partial\left(g_{c}^{*}\right)(\tau^{*})}{\partial c}<0$
for all $0<c<\tau^{*}$, fix each $0<c<\tau^*$. We now study  how $\frac{\partial\left(g_{c}^{*}\right)(\tau)}{\partial c}$
changes as a function of $\tau$. We can apply \citet[Theorem C.1, ][]{kitagawa2022treatment} to conclude that $\frac{\partial\left(g_{c}^{*}\right)(\tau)}{\partial c}$
has at most one sign change (as a function of $\tau$), as $\frac{\phi^{2}(y+c)\phi(y-c)}{\left(\phi(y+c)+\phi(y-c)\right)^{3}}y$
has one sign change from $-$ to $+$ as a function of $y$. Furthermore,
note 
\[
\frac{\partial\left(g_{c}^{*}\right)(\tau)}{\partial c}\mid_{\tau=c}=0,
\]
and we may verify 
\begin{align*}
\left(\frac{\partial\left(g_{c}^{*}\right)(\tau)}{\partial c\partial\tau}\right)\mid_{\tau=c} & =-4c^{2}\int\frac{\phi^{2}(y+c)\phi(y-c)}{\left(\phi(y+c)+\phi(y-c)\right)^{3}}y\phi(y-c)(y-c)dy\\
 & =-4c^{2}\int\frac{\phi^{2}(y+c)\phi(y-c)}{\left(\phi(y+c)+\phi(y-c)\right)^{3}}y^{2}\phi(y-c)dy<0,
\end{align*}
implying $\tau=c$ is indeed a point of sign change of $\frac{\partial\left(g_{c}^{*}\right)(\tau)}{\partial c}$.
Applying \citet[Theorem C.1, ][]{kitagawa2022treatment}, we conclude that $\frac{\partial\left(g_{c}^{*}\right)(\tau)}{\partial c}$ (as a function of $\tau$)
is first positive and then negative with one unique sign change at
$\tau=c$. As $\tau^{*}>c$, we conclude $\frac{\partial\left(g_{c}^{*}\right)(\tau^{*})}{\partial c}=\frac{\partial\left(g_{c}^{*}\right)(\tau)}{\partial c}\mid_{\tau=\tau^{*}}<0$
for all $0<c<\tau^{*}$. Statement (i) follows.

Statement (ii). The proof is similar. Viewing $\left(g_{s}^{*}\right)(c)$
as a function of $s$, we aim to show that $\frac{\partial\left(g_{s}^{*}\right)(c)}{\partial s}>0$
for all $c<s<\tau^{*}$. Algebra shows
\begin{align*}
\frac{\partial\left(g_{s}^{*}\right)(c)}{\partial s} & =-4c^{2}\int\frac{\phi^{2}(y+s)\phi(y-s)}{\left(\phi(y+s)+\phi(y-s)\right)^{3}}y\phi(y-c)dy.
\end{align*}
Now fix each $c<s<\tau^*$. Viewing $\frac{\partial\left(g_{s}^{*}\right)(c)}{\partial s}$
as a function of $c$, we can conclude that it has at most one sign
change by applying \citet[Theorem C.1, ][]{kitagawa2022treatment}. As
\[
\frac{\partial\left(g_{s}^{*}\right)(c)}{\partial s}\mid_{c=s}=0
\]
 and 
\begin{align*}
\frac{\partial\left(g_{s}^{*}\right)(c)}{\partial s\partial c}\mid_{c=s} & =-4s^{2}\int\frac{\phi^{2}(y+s)\phi(y-s)}{\left(\phi(y+s)+\phi(y-s)\right)^{3}}y\phi(y-s)(y-s)dy\\
 & =-4s^{2}\int\frac{\phi^{2}(y+s)\phi(y-s)}{\left(\phi(y+s)+\phi(y-s)\right)^{3}}y^{2}\phi(y-s)dy\\
 & <0.
\end{align*}
Therefore, $\frac{\partial\left(g_{s}^{*}\right)(c)}{\partial s}$
indeed has one unique sign change from positive to negative (as a function of $c$). As $\frac{\partial\left(g_{s}^{*}\right)(s)}{\partial s}=0$,
we conclude that 
\[
\frac{\partial\left(g_{s}^{*}\right)(c)}{\partial s}>0
\]
 for all $c<s<\tau^*$. Thus, statement (ii) follows.
\end{proof}

\bibliographystyle{ecta}
\bibliography{bibliography}

\end{document}